\RequirePackage{fix-cm}

\RequirePackage{amsmath}
\documentclass{llncs}

\usepackage[utf8]{inputenc}
\usepackage[T1]{fontenc}
\usepackage[english]{babel}
\usepackage{lmodern}

\usepackage{mathtools}
\usepackage{amssymb}
\usepackage{ebproof}
\usepackage{stmaryrd}

\usepackage{multicol}
\usepackage{tabularx}
\usepackage{multirow}

\usepackage{xspace}
\usepackage{csquotes}
\usepackage[style=base]{caption}
\usepackage{subcaption}
\usepackage{wrapfig}
\usepackage{microtype}
\usepackage{fnpct}
\usepackage{enumitem}

\usepackage{xcolor} %

\usepackage{tikz}
\usetikzlibrary{arrows, arrows.meta, calc, positioning, shapes, decorations.pathmorphing,fit}
\tikzset{>=Latex}

\usepackage{hyperref}

\addto\extrasenglish{

}

\usepackage{stackengine,scalerel,graphicx,trimclip,wasysym}
\savestack\zigzagtextstyle{\vphantom{()}\scalebox{.86666}[1]{\kern-.5pt\AC\kern-.5pt}}
\savestack\onelinetextstyle{\vphantom{()}\scalebox{1}[1]{\kern-1pt{$-$}\kern-1pt}}
\savestack\twolinetextstyle{\vphantom{()}\scalebox{1}[1]{\kern-1pt{$=$}\kern-1pt}}
\savestack\ZigZagtextstyle{
	$\vphantom{()}\smash{\raisebox{-.002em}{$\vcenter{\hbox{%
					\stackengine{-.805em}{\zigzagtextstyle}{\zigzagtextstyle}{O}{c}{F}{F}{S}}}$}}$}
\newlength\repwidth
\repwidth=\wd\zigzagtextstylecontent\relax
\newcommand\rightarrowhead{\clipbox{4pt -2pt 0pt -2pt}{$\rightarrow$}}
\newcommand\Rightarrowhead{\clipbox{4pt -2pt 0pt -2pt}{$\Rightarrow$}}
\newcommand\rrightarrowhead{\clipbox{5.5pt -2pt 0pt -2pt}{$\rightarrow\kern-8.5pt\rightarrow$}}
\newcommand\RRightarrowhead{\clipbox{5.5pt -2pt 0pt -2pt}{$\Rightarrow\kern-8.5pt\Rightarrow$}}
\newcommand\fXarrowtextstyle[2]{$\vphantom{()}\smash{\vcenter{\hbox{\kern-.03\repwidth%
				\clipbox{-.03\repwidth{} 0pt .527\repwidth{} 0pt}{#1}}}}#2$}
\savestack\fzigarrowtextstyle{\fXarrowtextstyle{\zigzagtextstyle}{\rightarrowhead}}
\savestack\fZigarrowtextstyle{\fXarrowtextstyle{\ZigZagtextstyle}{\Rightarrowhead}}
\savestack\flinearrowtextstyle{\fXarrowtextstyle{\onelinetextstyle}{\rightarrowhead}}
\savestack\fLinearrowtextstyle{\fXarrowtextstyle{\twolinetextstyle}{\Rightarrowhead}}
\savestack\fzzigarrowtextstyle{\fXarrowtextstyle{\zigzagtextstyle}{\rrightarrowhead}}
\savestack\fZZigarrowtextstyle{\fXarrowtextstyle{\ZigZagtextstyle}{\RRightarrowhead}}
\savestack\fllinearrowtextstyle{\fXarrowtextstyle{\onelinetextstyle}{\rrightarrowhead}}
\savestack\fLLinearrowtextstyle{\fXarrowtextstyle{\twolinetextstyle}{\RRightarrowhead}}

\newcommand\zigzag{\scalerel*{\zigzagtextstyle}{()}}

\newcommand\oneline{\scalerel*{\onelinetextstyle}{()}}

\newcommand\fzigarrow{\scalerel*{\fzigarrowtextstyle}{()}}

\newcommand\flinearrow{\scalerel*{\flinearrowtextstyle}{()}}

\newcommand\fllinearrow{\scalerel*{\fllinearrowtextstyle}{()}}

\newcommand\zigzagarrow[1][]{\Xarrow{\zigzag}{\fzigarrow}{#1}{-.65}}

\newcommand\linearrow[1][]{\Xarrow{\oneline}{\flinearrow}{#1}{-.65}}

\newcommand\llinearrow[1][]{\Xarrow{\oneline}{\fllinearrow}{#1}{-.65}}

\newcommand\Xarrow[4]{\ThisStyle{\mathrel{\Xarrowhelp{#1#2}{#3}{#1}{#4}}}}
\newcommand\Xarrowhelp[4]{%
	\setbox0=\hbox{$\SavedStyle#1$}%
	\setbox2=\hbox{$\SavedStyle_{\,\,#2\,\,}$}%
	\ifdim\wd0<\wd2\relax\Xarrowhelp{#3#1}{#2}{#3}{#4}%
	\else\stackengine{#4\LMex}{\copy0}{\copy2\,}{O}{c}{F}{T}{S}\fi%
}

\newcommand{\efs}[1][]{\linearrow[#1]}
\newcommand{\ebs}[1][]{\zigzagarrow[#1]}
\newcommand{\efbs}[1][]{\llinearrow[#1]}

\newcommand{\redl}[1]{\efs[#1]} %
\newcommand{\revredl}[1]{\ebs[#1]} %

\newcommand{\fwlts}[2]{\efs[#1:#2]}
\newcommand{\bwlts}[2]{\ebs[#1:#2]}
\newcommand{\fbwlts}[2]{\efbs[#1:#2]}
\newcommand{\tfbwlts}{\efbs[\star]}

\newcommand{\congru}{\equiv} %

\DeclareMathOperator{\co}{co}
\newcommand{\HPB}{\textnormal{HPB}\xspace}
\newcommand{\HHPB}{\textnormal{HHPB}\xspace}
\newcommand{\wfHPB}{\textit{wf}\HPB}
\newcommand{\wfHHPB}{\textit{wf}\HHPB}

\newcommand{\emptymem}{\emptyset}
\newcommand{\orig}[1]{O_{#1}} %
\newcommand{\fork}{\curlyvee} %
\newcommand{\bs}{\backslash}
\newcommand{\coh}{\frown}

\newcommand{\cat}[1]{\mathbb{#1}}
\newcommand{\catcs}{\cat{C}}
\newcommand{\catics}{\cat{D}}

\DeclareMathOperator{\ids}{\textrm{I}}
\DeclareMathOperator{\card}{Card}
\DeclareMathOperator{\id}{id}

\usepackage[OMLmathsfit]{isomath}
\newcommand{\iden}{\mathsfit{m}}

\newcommand{\iconf}{\ensuremath{\mathcal{D}}}

\DeclarePairedDelimiter{\encc}{\llbracket}{\rrbracket}
\newcommand{\enc}{\encc} %
\newcommand{\encr}{\encc} %
\newcommand{\names}{\ensuremath{\mathsf{N}}}
\newcommand{\proc}{\ensuremath{\mathsf{P}}}
\newcommand{\rproc}{\ensuremath{\mathsf{R}}}
\newcommand{\labels}{\ensuremath{\mathsf{L}}}
\newcommand{\polabel}{\ensuremath{\mathcal{L}}}
\newcommand{\iso}{\ensuremath{\cong}}

\newcommand{\out}[1]{\overline{#1}}
\newcommand{\power}{\ensuremath{\mathcal{P}}}
\newcommand{\conf}{\ensuremath{\mathcal{C}}}
\newcommand{\rel}{\ensuremath{\mathcal{R}}}
\DeclarePairedDelimiter{\mem}{\langle}{\rangle}
\newcommand{\labl}{\ell}
\newcommand{\relabl}{r}

\newcommand{\restr}[1]{\mathord{\upharpoonright_{#1}}}

\newcommand{\confzero}{\mathbf{0}}
\newcommand{\iconfzero}{\mathbf{0}}

\newcommand{\nat}{\mathbb{N}}
\newcommand{\integer}{\nat}
\newcommand{\setst}{ \mid } %

\newcommand{\st}{s.t.\ } %

\newcommand*{\resp}{resp.\@\xspace}
\newcommand{\BNFsepa}{\enspace \Arrowvert \enspace}

\newcommand{\functoric}{\ensuremath{\mathcal{F}}} %
\newcommand{\functorci}{\ensuremath{\mathcal{S}}} %
\DeclarePairedDelimiter{\encm}{\lfloor}{\rfloor}
\newcommand{\namelist}[1]{\overrightarrow{#1}}

\newcommand{\xmax}{\ensuremath{x_{\textsc{max}}}}
\newcommand{\te}[1]{\ensuremath{\textsc{te}_{#1}}}
\newcommand{\mt}[1]{\ensuremath{\textsc{mt}_{#1}}}
\newcommand{\me}[1]{\ensuremath{\textsc{me}_{#1}}}

\newcommand\hypo{\Hypo}
\newcommand\infer{\Infer}

\hypersetup{
	plainpages=false,
	bookmarksnumbered=true,
	bookmarksopen=true,
	bookmarksdepth=2,
	breaklinks=true,
	colorlinks=true,
	linkcolor=blue!50!red,
	urlcolor=green!70!black,
	pdfdisplaydoctitle= true,
	pdfinfo={
 Author={Clément Aubert, Ioana Cristescu},
 Title={History-Preserving Bisimulations on Reversible Calculus of Communicating Systems},
 Keywords={Formal semantics, Process algebras and calculi, Reversible CCS, Hereditary history-preserving bisimulation},
 Creator={PDFLaTeX},
 Producer={PDFLaTeX},
 Lang={en}
	}
}

\title{History-Preserving Bisimulations on Reversible Calculus of Communicating Systems%
}%

\titlerunning{History-Preserving Bisimulations on RCCS}

\author{%
 Clément Aubert\inst{1}
 \and
 Ioana Cristescu\inst{2}}
\authorrunning{C. Aubert \& I. Cristescu}

\institute{%
 Augusta University, Augusta, GA 30912, USA%
 \and%
 Inria Rennes, France%
}

\begin{document}

\maketitle

\begin{abstract}
	\emph{History-} and \emph{hereditary history-preserving bisimulation (\HPB and \HHPB)} are equivalences relations for denotational models of concurrency.
	Finding their counterpart in process algebras is an open problem, with some partial successes:
	there exists in calculus of communicating systems (CCS) an equivalence based on causal trees that corresponds to \HPB.
	In Reversible CSS (RCCS), there is a bisimulation that corresponds to \HHPB, but it considers only processes without \emph{auto-concurrency}.
	We propose equivalences on CCS with auto-concurrency that correspond to \HPB and \HHPB, and their so-called \enquote{weak} variants.
	The equivalences exploit not only reversibility but also the memory mechanism of RCCS.
\end{abstract}

\section{Introduction}

\paragraph{Reversing Concurrent Computation}

Implementing reversibility in a programming language often requires to record the history of the execution.
Ideally, this history should be \emph{complete}, so that every forward step can be unrolled, and \emph{minimal}, so that only the relevant information is saved.
Concurrent programming languages have a third requirement: the history should be \emph{distributed}, to avoid centralization of information. %
To fulfill those requirements, Reversible CCS~\cite{Danos2004,Danos2005} uses \emph{memories} attached to the threads of a process. %

\paragraph{Equivalences for Reversible Processes}
A theory of reversible concurrent computation relies not only on a syntax, but also on \enquote{meaningful} behavioral equivalences.
In this paper we study behavioral equivalences defined on configuration structures~\cite{Winskel1986}, which are denotational models for concurrency.
In configuration structures, an \emph{event} represents an execution step, and a \emph{configuration}---a set of events that occurred---represents a state.
A forward transition is then represented as moving from a configuration to one of its superset, and backward transitions have a \enquote{built-in} representation: it suffices to move from a configuration to one of its subset.
Many behavioral equivalences have been defined for configuration structures ; some of them, like \emph{history-} and \emph{hereditary history-preserving bisimulations} (\HPB and \HHPB), use that \enquote{built-in} notion of reversibility.

\paragraph{Encoding Reversible Processes in Configuration Structures}
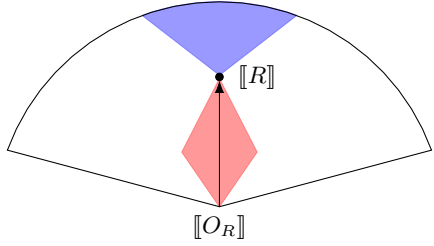
\begin{figure}
	{\centering
		\begin{tikzpicture}
			\node (P) at (0, 0){\(\enc{\orig{R}}\)};
			\node[inner sep = 0] (R) at (0,2) {};
			\node (Rt) at ($(R) + (0.5,0)$) {\(\enc{R}\)};
			\draw [fill=blue!40, draw=blue!50, domain=70:110] plot ({cos(\x) * 3}, {sin(\x)*3}) -- (R.north) -- cycle;
			\draw [fill= red!40, draw= red!50] (P.north) -- (-0.5,1) -- (R.south) -- (0.5,1) -- cycle;
			\draw [-Latex] (P.north) to ($(R.south) + (0, -0.02)$);
			\draw [domain=20:160] plot ({cos(\x) * 3}, {sin(\x)*3}) -- (P.north) -- cycle;
			\node[draw, shape=circle, inner sep = 0, minimum size=.1cm, fill=black] (R) at (0,2) {};
		\end{tikzpicture}

	}
	\caption{The encoding of \(\orig{R}\), of the future and of the past of a reversible \(R\).}
	\label{fig:schema}
\end{figure}
An ongoing research effort~\cite{Aubert2016jlamp,Phillips2011} is to transfer equivalences defined in denotational models, which are by construction adapted for reversibility, back into the reversible process algebra.
Of course, showing that an equivalence on configuration structures corresponds to one on RCCS processes depends on the encoding of RCCS terms into configuration structures.
One of them uses the fact that we are typically interested only in \emph{reachable} reversible processes---processes \(R\) which can backtrack to a process \(\orig{R}\) with an empty memory--- called its \emph{origin}.
Then, a natural choice is to consider \(\encr{\orig{R}}\)---the encoding of the origin of \(R\), using the common mapping for CCS processes~\cite{Winskel1986}---, and to identify in it the configuration corresponding to the current state of the reversible process.
In this set-up, the encoding of \(R\) is \emph{one configuration}, \(\enc{R}\), in the configuration structure \(\encr{\orig{R}}\): every configuration \enquote{below} is the \enquote{past} of \(R\), every configuration \enquote{above}, its \enquote{future} (\autoref{fig:schema}).

\paragraph{Contribution}

This paper improves %
on previous results by defining relations on CCS processes that correspond to \HPB, \HHPB, and their \enquote{weak function} variants.
The result does not require to consider a restricted class of processes.
We introduce an encoding of memories independent of the rest of the process and show that, as expected, the \enquote{past} of a process corresponds to the encoding of its memory.
The memories attached to a process are no longer only a syntactic layer to implement reversibility, but become essential for defining equivalences. This result gives an insight on the expressiveness of reversibility, as the back-and-forth moves of a process are not enough to capture \HHPB.

\paragraph{Related work}
The correspondence between \HHPB and back-and-forth bisimulations for processes without auto-concurrency~\cite{Aubert2016jlamp,Phillips2011} motivated some of the work presented here.
Our approach shares similarity with causal trees---in the sense that we encode only \emph{part of the execution} in a denotational representation---where some bisimulations corresponds to \HPB~\cite{Darondeau1990a}.

\paragraph{Outline}
We start by recalling the definitions of configuration structures (\autoref{sec:cs}), of the encoding of CCS in configuration structures (\autoref{sec:css-enc}), of (hereditary) history-preserving bisimulations (\autoref{sec:hpbs}), of RCCS (\autoref{sec:RCCS}) and of related notions. We also recall previous result on \HHPB (\autoref{thm:previousresult}).
We consider the reader familiar with CCS, in particular with its congruence relations and reduction rules.%

\autoref{sec:lifting} starts by defining a structure slightly richer than configuration structures, that we call \enquote{identified configuration structures} (\autoref{sec:ics}), and defines basic operations on them.
\autoref{sec:enc_mem} defines and illustrates with numerous examples how identified configuration structures can encode memories. %
Finally, \autoref{sec:equiv_rccs} uses this encoding to define relations on RCCS and CSS processes that are then stated to correspond to \HPB and \HHPB on configuration structures.

\autoref{sec:conclusion} concludes, and \autoref{sec:appendix} gathers the proofs and establishes the robustness of the tools introduced.

\section{Preliminary Definitions}
\label{sec:prelim}
We recall the definitions of configuration structures, auto-concurrency (\autoref{sec:cs}), how to encode CCS processes into configuration structures (\autoref{sec:css-enc}) and the history-preserving bisimulations (\autoref{sec:hpbs}).

We write \(\subseteq\) the set inclusion, \(\mathcal{P}\) the power set, \(\bs\) the set difference, \(\card\) the cardinal, \(\circ\) the composition of functions, \(A \to B\) the set of functions between \(A\) and \(B\), \(A \rightharpoonup B\) the partial functions and \(f \restr{C}\) the restriction of \(f : A \to B\) to \(C \subseteq A\).

Let \(\names=\{a,b,c,\dots\}\) be a set of \emph{names} and \(\out{\names}=\{\out{a},\out{b},\out{c},\dots\}\) its \emph{co-names}.
The complement of a (co-)name is given by a bijection \(\out{[\cdot]}:\names \to \out{\names}\), whose inverse is also denoted by \(\out{[\cdot]}\). %
We write \(\namelist{a}\) for a list of names \(a_1,\cdots,a_n\).
We define the sets of labels \(\labels = \names \cup \out{\names} \cup \{\tau\}\), let\(\polabel = \mathcal{P}(\labels)\), and use \(\alpha\) (resp.\ \(\lambda, \nu\)) to range over \(\labels\) (resp.\ \(\labels \bs \{\tau\}\)).

\subsection{Configuration Structures}
\label{sec:cs}

\begin{definition}[Configuration structures]
	\label{def:conf_str}
	A \emph{configuration structure} \(\conf\) is a tuple \((E, C, L, \labl)\) where \(E\) is a set of events, \(L\subseteq\polabel\) is a set of labels, \(\labl : E \to L \) is a labeling function and \(C \subseteq \power(E)\) is a set of subsets satisfying:
	\begin{align*}
		 & \forall x \in C, \forall e \in x, \exists z \in C\text{ finite}, e\in z, z \subseteq x \tag{Finiteness}                                   \\
		 & \forall x \in C, \forall e, e' \in x, e\neq e'\Rightarrow
		\exists z \in C, z \subseteq x,%
		e\in z \iff e'\notin z%
		\tag{Coincidence Freeness}                                                                                                                   \\
		 & \forall X \subseteq C, \exists y\in C \text{ finite }, \forall x \in X, x\subseteq y \Rightarrow \bigcup X \in C \tag{Finite Completness} \\
		 & \forall x,y\in C, x\cup y\in C \Rightarrow x\cap y\in C \tag{Stability}
	\end{align*}
	We denote \(\confzero\) the configuration structure with \(E = \emptyset\), and write \(x \redl{e} y\) and \(y \revredl{e}x\) for \(x, y \in C\) such that \(x=y\cup\{e\}\).
\end{definition}

For the rest of this paper, we often omit \(L\), %
let \(x,y,z\) range over configurations%
, and assume that we are always given \(\conf = (E, C, \labl)\) and \(\conf_i = (E_i, C_i, \labl_i)\), for \(i = 1, 2\).

\begin{definition}[Causality, Concurrency, and Maximality]
	\label{def:causality}
	For %
	\(x \in C\) and \(d, e \in x\), the \emph{causality relation on \(x\)} is given by \(d <_{x} e\) iff \(d \leqslant_x e\) and \(d \neq e\), where \(d \leqslant_x e\) iff for all \(y \in C\) with \(y \subseteq x\), we have \(e \in y \Rightarrow d \in y\).
	The \emph{concurrency relation on \(x\)}~\cite[Definition 5.6]{Glabbeek2001} %
	is given by \(d \co_x e\) iff \(\neg (d <_x e \vee e <_x d) \).
	Finally, \(x\) is \emph{maximal} if \(\forall y\in\conf\), \(x=y\) or \(x\not\subseteq y\).
\end{definition}

\begin{figure}
	{\centering
		\begin{minipage}[b]{.26\linewidth}
			\begin{tikzpicture}[scale=0.75]
				\node (emptyset) at (0,-1.5) {\(\emptyset\)};
				\node (a1) at (-1.5,0) {\(\{a_1\}\)};
				\node (a2) at (1.5,0) {\(\{a_2\}\)};
				\draw [->] (emptyset) -- (a1);
				\draw [->] (emptyset) -- (a2);
			\end{tikzpicture}
			\subcaption{\(\enc{a+a}\)}\label{fig:ex:3}
		\end{minipage}
		\hfill
		\begin{minipage}[b]{.26\linewidth}
			\begin{tikzpicture}[scale=0.75]
				\node (emptyset) at (0,-1.5) {\(\emptyset\)};
				\node (a1) at (-1.5,0) {\(\{a_1\}\)};
				\node (a2) at (1.5,0) {\(\{a_2\}\)};
				\node (a1a2) at (0, 1.5) {\(\{a_1, a_2\}\)};
				\draw [->] (emptyset) -- (a1);
				\draw [->] (emptyset) -- (a2);
				\draw [->] (a1) -- (a1a2);
				\draw [->] (a2) -- (a1a2);
			\end{tikzpicture}
			\subcaption{\(\enc{a \mid a}\)}\label{fig:ex:2}
		\end{minipage}
		\hfill
		\begin{minipage}[b]{.45\linewidth}
			\begin{tikzpicture}[scale=0.75]
				\node (emptyset) at (1,-1.5) {\(\emptyset\)};
				\node (b) at (-1,0) {\(\{a\}\)};
				\node (c) at (1,0) {\(\{\out{a}\}\)};
				\node (bcp) at (0, 1.5) {\(\{a, \out{a}\}\)};
				\node (ba) at (-2,1.5) {\(\{a, b\}\)};
				\node[align=center] (bapc) at (0,3) {\(\{a, \out{a}, b\}\)};
				\draw [->] (emptyset) -- (b);
				\draw [->] (emptyset) -- (c);
				\draw [->] (b) -- (bcp);
				\draw [->] (c) -- (bcp);
				\draw [->] (b) -- (ba);
				\draw [->] (ba) -- (bapc);
				\draw [->] (bcp) -- (bapc);
				\node (t) at (3, 0) {\(\{\tau\}\)};
				\node (tb) at (3, 1.5) {\(\{\tau, b\}\)};
				\draw [->] (emptyset) to (t);
				\draw [->] (t) to (tb);
			\end{tikzpicture}
			\subcaption{\(\enc{\out{a} \mid a.b}\)}\label{fig:ex:4}
		\end{minipage}
		\caption{Examples of configuration structures}
		\label{fig:ex}
	}
\end{figure}
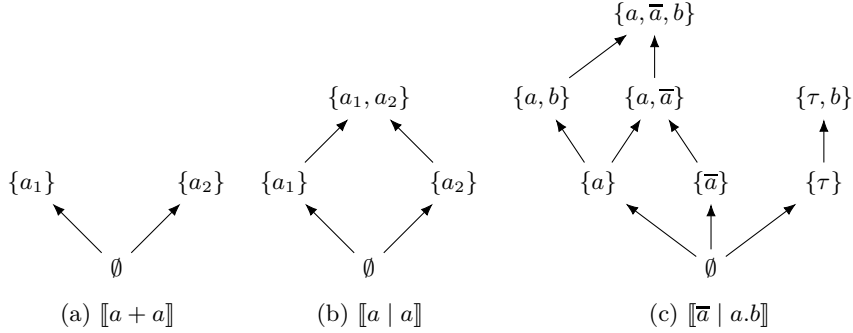

\begin{example}
	Consider the configuration structures of \autoref{fig:ex}, where the set of events and of configurations can be read from the diagram, and where we make the abuse of notation of writing the events as their labels (with a subscript if multiple events have the same label). %
	Note that two events with complement names can happen at the same time (\autoref{fig:ex:4}), in which case they are labeled with \(\tau\) and called \emph{silent transition}, as it is usual in CCS (\autoref{sec:css-enc}).
\end{example}

\begin{definition}[Category of configuration structures]
	\label{def:cat_lcs}
	We define \(\catcs\) the category of configuration structure, where objects are configuration structures, and a morphism \(f\) from \(\conf_1\) %
	to \(\conf_2\) %
	is a triple \((f_E, f_L, f_C)\) such that
	\begin{itemize}
		\item \(f_E : E_1 \to E_2\) preserves labels:
		      \(\labl_2(f_E(e)) = f_L(\labl_1(e))\), for \(f_L : L_1\to L_2\);
		\item \(f_C : C_1 \to C_2\) is defined as
		      \(f_C(x) = \{ f_E(e) : e\in x\}\).
	\end{itemize}
	If there exists an isomorphism \(f : \conf_1 \to \conf_2\), then we write \(\conf_1 \iso \conf_2\).
\end{definition}
We omit the \(f_L\) part of the morphisms when it is the identity morphism.

We now recall how process algebra constructors are defined on configuration structures~\cite{Winskel1982}.
The definition below may seem technical, but \autoref{def:encoding-CCS-cf} should make it clear that they capture %
the right notion.

This definition uses the product \((\times_{\star}, p_1, p_2)\) of the category of sets and partial functions~\cite[Appendix A]{Winskel1982}%
: letting \(\star\) denote \emph{undefined} for a partial function, \(C^{\star}= C \cup \{\star\}\) for a set \(C\), we define, for two sets \(A\) and \(B\),
\[
	A \times_{\star} B =
	\{(a,\star)\mid a \in A\}\cup\{(\star, b)\mid b \in B\}
	\cup\{(a, b)\mid a \in A, b\in B\}
\]
with \(p_1: A \times_{\star} B \to A^{\star}\) and \(p_2 : A \times_{\star} B \to B^{\star}\).

\begin{definition}[Operations on configuration structures~\cite{Aubert2016jlamp,Sassone1996}]
	\label{cat-op-def}
	\begin{description}[labelsep=.3333em]
		\item[The product]
		      of \(\conf_1\) and \(\conf_2\) is \(\conf_1\times\conf_2 = (E_1 \times_{\star} E_2, C, \labl)\).
		      Define the projections \(\pi_i : \conf \to \conf_i\) %
		      and the configurations \(x\in C\) such that:%
		      \begin{align*}
			       & \forall e \in E, \pi_i(e)=p_i(e),                                                                     %
			      \pi_i(\labl_i(e)) = \labl_i(\pi_i(e))                                                                    \\
			       & \pi_i(x)\in C_i                                                                                       \\
			       & \forall e,e'\in x, \pi_1(e)=\pi_1(e')\neq\star\text{ or }\pi_2(e)=\pi_2(e')\neq\star \Rightarrow e=e' \\
			       & \forall e \in x, \exists z\subseteq x \text{ finite},\pi_1(x) \in C_1, \pi_2(x)\in C_2,               %
			      e\in z                                                                                                   \\
			       & \forall e, e' \in x, e\neq e'\Rightarrow \exists z\subseteq x,                                        %
			      \pi_i(z) \in C_i,%
			      e\in z \iff e'\notin z%
		      \end{align*}

		      The labeling function \(\labl:E_1\times_{\star} E_2 \to L_1\cup L_2 \cup (L_1\times L_2)\) is %
		      \[
			      \labl(e) = \begin{dcases*}
				      \labl_1(e_1) & if \(\pi_1(e)=e_1 \neq\star\) and \(\pi_2(e)=\star\)\\
				      \labl_2(e_2) & if \(\pi_1(e)=\star\) and \(\pi_2(e)=e_2 \neq\star\)\\
				      (\labl_1(e_1),\labl_2(e_2)) & otherwise
			      \end{dcases*}
		      \]
		\item[The relabeling] of \(\conf_1\) along \(\relabl: E_1 \to L\) %
		      is \(\relabl \circ \conf_1 = (E_1, C_1, \relabl)\).

		\item[The restriction] of \(\conf_1\) to \(E \subseteq E_1\) is \(\conf_1%
		      \restr{E} = (E, C, \labl \restr{E})\), where \(x \in C \iff x \in C_1\) and \(x \subseteq E\). %
		      The restriction of \(\conf_1\) to a name \(a\) is \(\conf_1%
		      \restr{a} \coloneqq \conf_1%
		      \restr{E_1^a}\) where \(E_1^a=\{e \in E_1 \setst \labl(e) \notin \{a, \out{a}\}\}\).
		      For \(\namelist{a} = a_1, \hdots, a_n\) a list of names, we define similarly \(\conf_1 \restr{\namelist{a}} = \conf_1 \restr{E_1^{\namelist{a}}}\) for \(E_1^{\namelist{a}} = \{e \in E_1 \setst \labl(e) \notin \{a_1, \out{a_1}, \hdots, a_n, \out{a_n}\}\}\).

		\item[The parallel composition] of \(\conf_1\) and \(\conf_2\) is \(\conf_1 \mid \conf_2 = \big(\relabl \circ (\conf_1\times\conf_2)\big)\restr{F}\),
		      with %
		      \begin{itemize}
			      \item \(\conf_1\times\conf_2 = \conf_3 = (E_3,C_3,\labl_3)\) is the product;
			      \item \(\relabl \circ \conf_3\) with \(\relabl: E_3 \to L_1\cup L_2\cup\{\bot\} \) defined as follows:
			            \[
				            \relabl(e) = \begin{dcases*}
					            \labl_3(e) & if \(\labl_3(e)\in\{a,\out{a}\}\)\\
					            \tau & if \(\labl_3(e)\in\{(a,\out{a}), (\out{a},a), \tau\}\)\\
					            \bot & otherwise
				            \end{dcases*}
			            \]
			      \item \((\relabl \circ \conf_3)\restr{F}\)%
			            , where \(F=\{e\in E_3 \setst \relabl(e)\neq \bot\}\).
		      \end{itemize}

		\item[The coproduct] of \(\conf_1\) and \(\conf_2\) is \(\conf_1+\conf_2 = \conf\), where \(E=(\{1\}\times E_1)\cup(\{2\}\times E_2)\) and \(C=\{\{1\}\times x \setst x\in C_1\}\cup\{\{2\}\times x \setst x\in C_2\}\).
		      The labeling function \(\labl\) is defined as \(\labl(e)=\labl_i(\pi_2(e))\) when \(\pi_1(e) = i\).%

		\item[The prefixing] of \(\conf_1\) by the name \(\lambda\) is \(\lambda . \conf_1 = (e \cup E_1, C, \labl)\), for \(e \notin E_1\) where \(x \in C \iff x = \emptyset \vee \exists x' \in C_1\), \(x = x' \cup e\) ; \(\labl(e) = \lambda\) and \(\forall e' \neq e\), \(\labl(e') = \labl_1(e')\).
	\end{description}
\end{definition}

\begin{definition}[Auto-concurrency~{\cite[Definition 9.5]{Glabbeek2001}}]
	\label{def:autoc}
	If \(\forall x \in C\), \(\forall e_1, e_2 \in x\), \(e_1 \co_x e_2\) and \(\labl(e_1) = \labl(e_2)\) implies \(e_1 = e_2\), then \(\conf\) is \emph{without auto-concurrency}.
\end{definition}

Any configuration structure where configurations have at most one event (like \autoref{fig:ex:3}) are without auto-concurrency.
\autoref{fig:ex:2}%
, on the other hand, is a configuration structure with auto-concurrency:
for \(x = \{a_1, a_2\}\) we have that \(a_1 \co_{x} a_2\),
\(\labl(a_1) = \labl(a_2)\), and yet \(a_1 \neq a_2\).

\subsection{CCS and its Encoding in Configuration Structures}
\label{sec:css-enc}
The set of CCS processes \(\proc\) is inductively defined:
\begin{align}
	P,Q & \coloneqq \lambda . P \BNFsepa P\mid Q \BNFsepa \lambda . P+ \nu . Q \BNFsepa P\bs a \BNFsepa 0 \tag{CCS Processes}
\end{align}

In the category of configuration structures \(\catcs\) (\autoref{def:cat_lcs}) one can \enquote{match} the process constructors of CCS with the categorical operations of \autoref{cat-op-def}.

\begin{definition}[{Encoding a CCS process~\cite[p.~57]{Winskel1995}}]
	\label{def:encoding-CCS-cf}
	Given a CCS process \(P\), its encoding \(\enc{P}\) as a configuration structure is built inductively:%
	\begin{align*}
		\enc{\lambda . P}         & =\lambda . \enc{P}                                                                        %
		                          &                                &  & \enc{P\mid Q} & =\enc{P}\mid\enc{Q} &            &  &
		\enc{\lambda .P + \nu .Q} & =\enc{\lambda .P}+\enc{\nu .Q}                                                            \\
		\enc{P \bs a}             & =\enc{P}\restr{a}              &  &               & \enc{0}             & =\confzero
	\end{align*}
\end{definition}
For now on we assume that all structures use the same set of labels \(\labels\).

\begin{definition}[Auto-concurrency in CCS]
	\label{def:autocrccs}
	A process \(P\) is \emph{without auto-concurrency} %
	if \(\enc{P}\) is.
\end{definition}

\subsection{(Hereditary) History-Preserving Bisimulations}
\label{sec:hpbs}
HPB~\cite{Rabinovich1988,Phillips2011}, {\cite[Theorem~4]{Baldan2010}} %
and \HHPB~{\cite[Definition 1.4]{Bednarczyk1991}, {\cite[Theorem~1]{Baldan2010}}} %
are equivalences on configuration structures that use label- and order-preserving bijections between the events of the two configuration structures.

\begin{definition}[Label- and order-preserving functions]%
	A function \(f : x_1 \to x_2\), for \(x_i\in C_i\), \(i\in\{1,2\}\) is label-preserving if \(\labl_1(e) = \labl_2(f(e))\) for all \(e \in x_1\).
	It is \emph{order-preserving} if \(e_1 \leqslant_{x_1} e_2 \Rightarrow f(e_1) \leqslant_{x_2} f(e_2)\), for all \(e_1, e_2 \in x_1\). %
\end{definition}

\begin{definition}[\HPB and \HHPB]
	\label{def:hpb}
	\label{def:hhpb}
	A relation \(\rel\subseteq C_1\times C_2\times (E_1 \rightharpoonup E_2)\) %
	such that \((\emptyset,\emptyset,\emptyset)\in {\rel}\), and if \((x_1,x_2,f)\in {\rel}\), then \(f\) is a label- and order-preserving bijection between \(x_1\) and \(x_2\) and (\ref{hpb1}) and (\ref{hpb2}) (resp.\ (\ref{hpb1}--\ref{hpb4})) hold is called a \emph{history-} (resp.\ \emph{hereditary history-}) \emph{preserving bisimulation between \(\conf_1\) and \(\conf_2\)}.
	\begin{align}
		\forall y_1, x_1\redl{e_1}y_1 \Rightarrow
		\exists y_2, g, x_2\redl{e_2}y_2, g\restr{x_1} = f, (y_1, y_2, g)\in {\rel} \label{hpb1}     \\
		\forall y_2, x_2\redl{e_2}y_2 \Rightarrow
		\exists y_1, g, x_1\redl{e_1}y_1', g\restr{x_1} = f, (y_1, y_2, g)\in {\rel} \label{hpb2}    \\
		\forall y_1, x_1 \revredl{e_1} y_1 \Rightarrow
		\exists y_2, g, x_2 \revredl{e_2} y_2, g = f\restr{x_1}, (y_1, y_2, g)\in {\rel}\label{hpb3} \\
		\forall y_2, x_2 \revredl{e_2} y_2 \Rightarrow
		\exists y_1, g, x_1 \revredl{e_1} y_1, g = f\restr{x_1}, (y_1, y_2, g)\in {\rel} \label{hpb4}
	\end{align}
\end{definition}

Note that the bijection on events is preserved from one step to the next.
This condition can be weakened, and we call the corresponding relations the \emph{weak-function} \HPB and \emph{weak-function} \HHPB~{\cite[Definition1.4]{Bednarczyk1991}}, {\cite[Definition3.11]{Phillips2011}}%
\footnote{The names \emph{weak}-HPB and \emph{weak}-\HHPB are more common~\cite{Bednarczyk1991,Phillips2011}, but can be confused with the \emph{weak} equivalences of process algebra, which refers to ignoring \(\tau\)-transitions.}.

\begin{definition}[\wfHPB]
	\label{def:whpb}
	A \emph{weak-function history-preserving bisimulation between \(\conf_1\) and \(\conf_2\)} is a relation \(\rel\subseteq C_1\times C_2\times (E_1 \rightharpoonup E_2 )\) such that \((\emptyset,\emptyset,\emptyset)\in {\rel}\) and if \((x_1,x_2,f)\in {\rel}\), then \(f\) is a label- and order-preserving bijection between \(x_1\) and \(x_2\) and
	\begin{align*}
		\forall y_1, x_1\redl{e_1}y_1 \Rightarrow
		\exists y_2, g, x_2\redl{e_2}y_2, (y_1,y_2,g)\in {\rel} \\
		\forall y_2, x_2\redl{e_2}y_2 \Rightarrow
		\exists y_1, g, x_1\redl{e_1}y_1, (y_1,y_2,g)\in {\rel}%
	\end{align*}
\end{definition}

Similarly one defines \wfHHPB.
If there is a \HPB %
between \(\conf_1\) and \(\conf_2\), we %
just write that \(\conf_1\) and \(\conf_2\) are \HPB, and similarly for \HHPB, \wfHPB and \wfHHPB.

\begin{figure}
	\hspace{4em}
	\begin{minipage}[b]{.3\linewidth}
		\begin{tikzpicture}[scale=0.75]
			\node (emptyset) at (1,-1.5) {\(\emptyset\)};
			\node (a) at (-0.5,0) {\(\{a_1\}\)};
			\node (b) at (1,0) {\(\{b\}\)};
			\node (ab) at (1, 1.5) {\(\{a_1, b\}\)};
			\node (aa) at (-0.5,1.5) {\(\{a_1, a_2\}\)};
			\node[align=center] (aab) at (1,3) {\(\{a_1, a_2, b\}\)};
			\draw [->] (emptyset) -- (a);
			\draw [->] (emptyset) -- (b);
			\draw [->] (a) -- (aa);
			\draw [->] (a) -- (ab);
			\draw [->] (b) -- (ab);
			\draw [->] (b) -- (ab);
			\draw [->] (ab) -- (aab);
			\draw [->] (aa) -- (aab);
		\end{tikzpicture}
	\end{minipage}
	\hfill
	\begin{minipage}[b]{.35\linewidth}
		\begin{tikzpicture}[scale=0.75]
			\node (emptyset) at (1,-1.5) {\(\emptyset\)};
			\node (a1) at (-1,0) {\(\{a_1\}\)};
			\node (a2) at (1,0) {\(\{a_2\}\)};
			\node (b) at (3, 0) {\(\{b\}\)};
			\node (aa) at (-1, 1.5) {\(\{a_1,a_2\}\)};
			\node (a1b) at (1,1.5) {\(\{a_1, b\}\)};
			\node (a2b) at (3, 1.5) {\(\{a_2, b\}\)};
			\node[align=center] (aab) at (1,3) {\(\{a_1, a_2, b\}\)};
			\draw [->] (emptyset) -- (a1);
			\draw [->] (emptyset) -- (a2);
			\draw [->] (emptyset) -- (b);
			\draw [->] (a1) -- (a1b);
			\draw [->] (a2) -- (a2b);
			\draw [->] (a1) -- (aa);
			\draw [->] (a2) -- (aa);
			\draw [->] (b) -- (a1b);
			\draw [->] (b) -- (a2b);
			\draw [->] (a1b) -- (aab);
			\draw [->] (a2b) -- (aab);
			\draw [->] (aa) -- (aab);
		\end{tikzpicture}
	\end{minipage}
	\hspace{4em}
	\caption{There is no \HHPB relation between \(\enc{a.a \mid b}\) and \(\enc{a \mid a \mid b}\)%
		.}\label{fig:not_in_hhpb}
\end{figure}

\begin{example}
	\label{ex:a.a|b}
	Observe that \(\enc{a.a \mid b}\) and \(\enc{a \mid a \mid b}\), presented in \autoref{fig:not_in_hhpb}, are \HPB, but not \HHPB.
	Any \HHPB relation would have to associate the maximal configurations \(\{a_1, a_2, b\}\) of the two structures and to construct a bijection: taking \((a_1 \mapsto a_1, a_2 \mapsto a_2, b \mapsto b)\) wouldn't work, since \(\enc{a \mid a \mid b}\) can backtrack on \(a_2\) and \(\enc{a.a \mid b}\) cannot.
	Taking the other bijection, \((a_1 \mapsto a_2, a_2 \mapsto a_1, b \mapsto b)\), fails too, since \(\enc{a \mid a \mid b}\) can backtrack on \(a_1\) and \(\enc{a.a \mid b}\) cannot.

\end{example}

\begin{example}[\cite{Glabbeek2001}]
	The processes \(a\mid (b+c) + a\mid b + (a+c) \mid b\) and \(a\mid(b+c) + (a+c)\mid b\) is another example of processes whose encodings are \HPB but not \HHPB.
\end{example}

\begin{example}
	Finally, observe that \(a.(b+b)\) and \(a.b+a.b\) are not structuraly congruent in CCS, but the encoding of the two processes are \HHPB.
\end{example}

\subsection{Reversible CCS and Coherent Memories}
\label{sec:RCCS}
Let \(\ids\) be a set of identifiers, and \(i, j, k\) range over elements of \(\ids\).
The set of RCCS processes \(\rproc\) is built on top of the set of CCS processes \(\proc\) (\autoref{sec:cs}):
\begin{align}
	e    & \coloneqq \mem{i, \lambda, P}                                            & \tag{Memory Events} \\
	m    & \coloneqq \emptymem \BNFsepa \fork . m \BNFsepa e. m \tag{Memory Stacks}                       \\
	T, U & \coloneqq m \rhd P \tag{Reversible Thread}                                                     \\
	R, S & \coloneqq T \BNFsepa T \mid U \BNFsepa R\bs a \tag{RCCS Processes}
\end{align}

We denote \(\ids(m)\) (\resp \(\ids(R)\), \(\ids(e)\)) the set of identifiers occurring in \(m\) (\resp \(R\), \(e\)), and always take
\(\ids = \integer\).
A structural congruence \(\congru\) can be defined on RCCS terms~\cite[Definition~5]{Aubert2016jlamp}, %
the only rule we will use here is \emph{distribution of memory}: \(m \rhd (P \mid Q) \congru (\fork . m \rhd P) \mid (\fork . m \rhd Q)\).
We also note that \(R\congru (\prod_{i} m_i\rhd P_i)\bs \namelist{a}\), for any reversible process \(R\) and for some names \(\namelist{a}\), memories \(m_i\) and CCS processes \(P_i\), writing \(\prod_{i}\) the \(i\)-ary parallel composition. %

\begin{figure}
	{
		\centering

		\begin{prooftree}
			\hypo{}
			\infer[left label=\(i \notin \ids(m)\)]1[act.]{(m \vartriangleright \lambda . P + Q) \fwlts{i}{\lambda} \mem{i, \lambda, Q} . m \vartriangleright P}
		\end{prooftree}
		\hfill
		\begin{prooftree}
			\hypo{R \fwlts{i}{\lambda} R'}
			\hypo{S \fwlts{i}{\out{\lambda}} S'}
			\infer2[syn.]{R \mid S \fwlts{i}{\tau} R' \mid S'}
		\end{prooftree}
		\\[2em]
		\begin{prooftree}
			\hypo{}
			\infer[left label=\(i \notin \ids(m)\)]1[act.\(_*\) ]{\mem{i, \lambda, Q} . m \vartriangleright P \bwlts{i}{\lambda} m \vartriangleright (\lambda . P + Q)}
		\end{prooftree}
		\hfill
		\begin{prooftree}
			\hypo{R \bwlts{i}{\lambda} R'}
			\hypo{S \bwlts{i}{\out{\lambda}} S'}
			\infer2[syn.\(_*\)]{R \mid S \bwlts{i}{\tau} R' \mid S'}
		\end{prooftree}
		\\[2em]
		\begin{prooftree}
			\hypo{R \fbwlts{i}{\alpha} R'}
			\infer[left label=\(i \notin \ids(S)\)]1[par.]{R \mid S \fbwlts{i}{\alpha} R' \mid S}
		\end{prooftree}
		\hfill
		\begin{prooftree}
			\hypo{R \fbwlts{i}{\alpha} R'}
			\hypo{a \notin \alpha}
			\infer2[res.]{ R\bs a \fbwlts{i}{\alpha} R'\bs a}
		\end{prooftree}
		\hfill
		\begin{prooftree}
			\hypo{R_1 \congru R \fbwlts{i}{\alpha} R' \congru R_1'}
			\infer1[\(\congru\)]{R_1 \fbwlts{i}{\alpha} R_1'}
		\end{prooftree}

	}
	\caption{Rules of the labeled transition system}\label{ltsrules}
\end{figure}

The labeled transition system for RCCS is given by the rules of \autoref{ltsrules}.
We use \(\fbwlts{i}{\alpha}\) as a wildcard for \(\redl{i:\alpha}\) (forward) or \(\revredl{i:\alpha}\) (backward transition), and if there are indices \(i_1, \hdots, i_n\) and labels \(\alpha_1, \hdots, \alpha_n\) such that \(R_1 \fbwlts{i_1}{\alpha_1} \cdots \fbwlts{i_n}{\alpha_n} R_n\), then we write \(R_1 \tfbwlts R_n\).
If there is a CCS term \(P\) such that \(\emptymem\rhd P\tfbwlts R\), we say that \(R\) is \emph{reachable}, that \(P\) is the \emph{origin} of \(R\)~\cite[Lemma 1]{Aubert2016jlamp} and write \(P = \orig{R}\).
Similarly to what we did in \autoref{def:autocrccs}, we will write that \(R\) is without auto-concurrency if \(\enc{\orig{R}}\) is.
An example of the execution of a reversible process is given at the beginning of \autoref{ex:proc}.

Note that we \emph{cannot} work up to \(\alpha\)-renaming of the identifiers: concurrent, distributed computation is about splitting threads between independent units of computation.
If one of unit were to re-tag a memory event \(\mem{89, \lambda, P}\) as \(\mem{1, \lambda, P}\), and another were to try to backtrack on the memory event \(\mem{89, \out\lambda, P}\), then the trace of the synchronization would be lost, and backtracking made impossible.
Since we don't want to keep a \enquote{global index} %
which would go against the benefits of distributed computation, the only %
option is to forbid \(\alpha\)-renaming of identifiers.

Memory coherence~\cite[Definition 1]{Danos2004} %
was defined for RCCS processes with less structured memory events (i.e., without identifiers), but can be adapted.

\begin{figure}

	{
		\centering

		\begin{prooftree}
			\hypo{}
			\infer1[em.]{\emptymem \coh \emptymem}
		\end{prooftree}
		\qquad
		\begin{prooftree}
			\hypo{m \coh m'}
			\infer[left label={\(\ids(e) \notin \ids(m) \cup \ids(m')\)}]1[ev.]{e.m \coh m'}
		\end{prooftree}
		\\[2em]
		\begin{prooftree}
			\hypo{m \coh m'}
			\infer[left label={\(i \notin \ids(m) \cup \ids(m')\)}]1[syn.]{\mem{i, \lambda, P}.m \coh \mem{i, \out{\lambda}, Q}.m'}
		\end{prooftree}
		\qquad
		\begin{prooftree}
			\hypo{m \coh \emptyset}
			\infer1[fo.]{\fork.m \coh \fork.m}
		\end{prooftree}

	}
	\caption{Rules for the coherence relation}
	\label{fig:coh-rules}
\end{figure}

\begin{definition}[Coherence relation]
	\label{def:coherence}
	Coherence, written \(\coh\), is the smallest symmetric relation on memory stacks such that %
	rules of \autoref{fig:coh-rules} hold.
\end{definition}

Note that \(\coh\) is not reflexive, and hence not an equivalence, nor anti-reflexive. %
For the rest of this paper, we will just write \enquote{memory} for \enquote{memory stack}.

\begin{definition}[Coherent processes, {\cite[Definition 2]{Danos2004}}]
	\label{def:coh}
	A RCCS process is \emph{coherent} if all of its memories are pairwise coherent, or if its only memory is coherent with \(\emptyset\).
\end{definition}

We require the memory to be coherent with \(\emptyset\) to make it impossible to have \(\mem{1, \lambda, P} . \mem{1, \nu, Q} . \emptyset\) as a memory in a coherent process.

\begin{lemma}[{\cite[Lemma 5]{Danos2005}}]
	\label{lem:reach-coh}
	If \(R \tfbwlts S\) and \(R\) is coherent then so is \(S\).
\end{lemma}

\begin{corollary}
	\label{cor:unique_id}
	For every reachable \(m\rhd P\) and \(i \in \ids(m)\), \(i\) occurs once in \(m\).
\end{corollary}

Note that the property above holds for reversible threads, %
and not for RCCS processes in general: indeed, we actually \emph{want} memory events to have the same identifiers if they result from a synchronization or a fork.

\begin{definition}[Back-and-forth bisimulation]
	\label{def:bfbism}
	A \emph{back-and-forth bisimulation} in RCCS is a relation \(\rel\subseteq \rproc \times \rproc\) such that if \((R_1, R_2)\in\rel\)
	\begin{align*}
		\forall S_1,R_1\fwlts{i}{\alpha}S_1\Rightarrow
		\exists S_2,R_2\fwlts{j}{\alpha}S_2, (S_1, S_2) \in\rel \\
		\forall S_2,R_2\fwlts{i}{\alpha}S_2\Rightarrow
		\exists S_1,R_1\fwlts{j}{\alpha}S_1, (S_1, S_2) \in\rel \\
		\forall S_1,R_1\bwlts{i}{\alpha}S_1\Rightarrow
		\exists S_2,R_2\bwlts{j}{\alpha}S_2, (S_1, S_2) \in\rel \\
		\forall S_2,R_2\bwlts{i}{\alpha}S_2\Rightarrow
		\exists S_1,R_1\bwlts{j}{\alpha}S_1, (S_1, S_2) \in\rel
	\end{align*}

\end{definition}

\begin{example}
	\label{ex:not_bf}
	The processes \(\emptymem\rhd (a.a \mid b)\) and \(\emptymem\rhd (a \mid a \mid b)\) are in a back-and-forth bisimulation.
\end{example}

\begin{theorem}[{\cite[Theorem 2]{Aubert2016jlamp}},\cite{Phillips2011}]
	\label{thm:previousresult}
	Back-and-forth bisimulation on RCCS processes without auto-concurrency corresponds to \HHPB on their encoding.
\end{theorem}

Note that this result does not hold for the processes in~\autoref{ex:not_bf} (see also~\autoref{ex:a.a|b}). This does not contradict the theorem, as the processes in~\autoref{ex:not_bf} are with auto-concurrency.

\section{Lifting the Restrictions}
\label{sec:lifting}
To define a bisimulation on RCCS (with auto-concurrency) that corresponds to \HHPB we first have to encode the memories of a reversible process into a structure similar to the configuration structures, called \emph{identified configuration structures} (\autoref{sec:ics}).
We can then define the encoding (\autoref{sec:enc_mem}), and the equivalences in RCCS that use this encoding of memories (\autoref{sec:equiv_rccs}).

\subsection{Identified Configuration Structures}
\label{sec:ics}
\begin{definition}[Identified configuration structure]\label{def:iconf}
	An \emph{identified configuration structure}, or \emph{\(\iconf\)-structure}, \(\iconf = (E, C, \labl, \ids, \iden)\) is a configuration structure \(\conf = (E, C , \labl)\) endowed with a set of identifiers \(\ids\) and a function \(\iden : E \to \ids\) such that, \begin{equation}
		\forall x \in C, \forall e_1, e_2 \in x, \iden(e_1) \neq \iden(e_2) \tag{Collision Freeness}\label{eq:collision}
	\end{equation}
	We call \(\conf\) \emph{the underlying configuration structure} of \(\iconf\) and write \(\iconf = \conf + \iden\). %
	We write \(\iconfzero\) for the identified configuration structure with \(E = \emptyset\).
\end{definition}

For the rest of this paper, we omit \(\ids\), %
and assume that we are always given \(\iconf = (E, C, \labl, \iden)\) and
\(\iconf_i = (E_i, C_i, \labl_i, \iden_i)\), for \(i = 1, 2\).

\begin{example}
	\autoref{fig:ex:4}, with \(\iden(a) = 1\), \(\iden(\out{a}) = 2\), \(\iden(b) = 3\) and \(\iden(\tau) = 4\), is a \(\iconf\)-structure.
	Note that it is possible to have %
	fewer identifiers than events: take \(\ids = \{1, 2, 3\}\) and \(\iden(a) = \iden(\tau) = 1\), \(\iden(\out{a}) = 2\) and \(\iden(b) = 3\).%
\end{example}

For the following remark, we need to suppose that every configuration structure is endowed with a total ordering on its events.

\begin{remark}
	\label{remark1}
	Every configuration structure can be mapped to a \(\iconf\)-structure.
\end{remark}

The mapping is trivial: take \(\ids\) to be \(\{1, \hdots, \card E\}\), and define \(\iden : E \to \ids\) to follow the ordering on \(E\).
Note that, in this case, \(\iden\) is a bijection.

\begin{definition}[Category of \(\iconf\)-structures]
	\label{def:cat_ics}
	We define \(\catics\) the category of identified configuration structure, where objects are \(\iconf\)-structures, and a morphism \(f\) from \(\iconf_1\) to \(\iconf_2\) is a triple \((f_E, f_C, f_{\iden})\) such that
	\begin{itemize}
		\item \((f_E, f_C)\) is a morphism in \(\catcs\) from \((E_1, C_1, \labl_1)\) to \((E_2, C_2, \labl_2)\);
		\item \(f_{\iden} : \ids_1 \to \ids_2\) preserves identifiers: \(f_{\iden}(\iden_1(e)) = \iden_2(f_E(e))\).
	\end{itemize}
\end{definition}

We denote \(\functoric : \catics \to \catcs\) the forgetful functor.

\begin{definition}[Operations on \(\iconf\)-structures]
	\label{icat-op-def}
	\begin{description}[labelsep=.3333em]
		\item[The product] of \(\iconf_1\) and \(\iconf_2\) is \(\iconf_1\times\iconf_2 = (E,C,\labl,\iden)\):
		      \begin{itemize}
			      \item \((E_1, C_1, \labl_1) \times (E_2,C_2,\labl_2) = (E,C,\labl)\) is
			            the product in the category of configuration structures with projections \(\pi_i:(E,C,\labl) \to (E_i,C_i,\labl_i)\);
			      \item \(\iden:E_1\times_{\star} E_2 \to (\ids_1\cup\{\iden_{\star}\})\times (\ids_2\cup\{\iden_{\star}\}) \), for \(\iden_{\star}\notin\ids_1\cup\ids_2\), is defined as

			            \[
				            \iden(e) = \begin{dcases*}
					            (\iden_1(\pi_1(e)),\iden_{\star}) & if \(\pi_2(e)=\star\)\\
					            (\iden_{\star},\iden_2(\pi_2(e))) & if \(\pi_1(e)=\star\)\\
					            (\iden_1(\pi_1(e)),\iden_2(\pi_2(e))) & otherwise
				            \end{dcases*}
			            \]
			            with the projections \(p_i:\iden\to\iden_i\cup\{\iden_{\star}\}\).
		      \end{itemize}
		      Define the projections \(\gamma_i : \iconf_1 \times \iconf_2 \to \iconf_i\) as the pair \((\pi_i,p_i)\).

		\item[The relabeling] of \(\iconf_1\) along \(\relabl: E_1 \to L\) is \(\relabl \circ \iconf_1 = (E_1, C_1, \relabl, \iden_1)\).

		\item[The restriction] of \(\iconf_1\) to \(E \subseteq E_1\) is \(\iconf_1 \restr{E} = \big((E_1, C_1, \labl_1)\restr{E} + \iden_1\restr{E}\big)\).%

		\item[The parallel composition] of \(\iconf_1\) and \(\iconf_2\) is \(\iconf_1 \mid \iconf_2 = \big( r \circ (\iconf_1\times \iconf_2)\big) \restr{F}\), with %
		      \begin{itemize}
			      \item \(\iconf_1\times \iconf_2 = \iconf_3 = (E_3, C_3, \labl_3, \iden_3)\) is the product of \(\iconf\)-structures.

			      \item \(\relabl \circ \iconf_3 \)%
			            , with \(\relabl : E_3 \to L \cup \{\bot\}\) defined as follows,%
			            \[
				            \relabl(e) = \begin{dcases*}
					            \bot
					            & if \(\pi_1(e)=e_1\neq\star \wedge \pi_2(e)=e_2\neq\star \wedge \iden_1(e_1) \neq \iden_2(e_2)\)\\
					            & or if \(\pi_1(e)=e_1\neq\star\wedge\pi_2(e)=\star\)\\
					            & \(\qquad \wedge (\exists e_2 \in E_2\), \st \(\iden_1(e_1) = \iden_2(e_2))\)\\
					            & or if \(\pi_2(e)=e_2\neq\star\wedge\pi_1(e)=\star\)\\
					            & \(\qquad\wedge (\exists e_1 \in E_1\), \st \(\iden_1(e_1) = \iden_2(e_2))\)\\
					            \tau
					            & if \(\pi_1(e)=e_1\neq\star \wedge \pi_2(e)=e_2\neq\star \wedge \iden_1(e_1)=\iden_2(e_2)\)\\
					            & and \(\labl_3(e)=(\alpha,\out\alpha)\) \hfill\emph{(Valid Synchronisations)}\\
					            \alpha
					            & if \(\pi_1(e)=e_1\neq\star \wedge \pi_2(e)=e_2\neq\star \wedge \iden_1(e_1)=\iden_2(e_2)\)\\
					            & and \(\labl_3(e)=(\alpha,\alpha)\) \hfill\emph{(Valid Forks)}\\
					            \labl_3(e) & otherwise
				            \end{dcases*}
			            \]%

			      \item \((r \circ \iconf_3)\restr{F}\), where %
			            \(F = \{e\in E_3 \mid r (e)\neq\bot\}\).
		      \end{itemize}

	\end{description}
\end{definition}

In the definition of parallel composition, \(\relabl\) detects the wrong synchronization (or fork) pairs: if two events occur at the same time, then they must have the same identifier.
And if an event occurs alone, then no other event can occur with the same identifier.

\subsection{Encoding the Memory of Reversible Processes}
\label{sec:enc_mem}
\begin{definition}[Encoding a RCCS memory]
	\label{def:encode_m}
	The encoding of the memory of a RCCS process in a \(\iconf\)-structure is defined\footnote{We make the abuse of notation of writing \(\encm{\cdot}\) for the encoding of both a reversible process and a memory, into a \(\iconf\)-structure.} by induction on the process:
	\begin{align*}
		\encm{m \rhd P}               & = \encm{m}                  &  &  &
		\encm{R_1 \mid R_2}           & = \encm{R_1} \mid\encm{R_2} &  &  &
		\encm{R\bs a}                 & = \encm{R}                          \\
		\encm{\mem{i, \alpha, P} . m} & = (E, C, \iden)             &  &  &
		\encm{\fork . m}              & = \encm{m}                  &  &  &
		\encm{\emptymem}              & = \confzero
	\end{align*}
	For \( \encm{\mem{i, \alpha, P} . m}\), letting \(\encm{m} = (E_m, C_m, \labl_m, \iden_m)\) and \(e \notin E_m\), we let \((E, C, \labl, \iden)\) be %
	\(E = E_m \cup e\), \(C = C_m \cup %
	\{x \cup \{e\} \setst x \in C_m \text{ is maximal}\}\), \( \labl = \labl_m \cup \{e \mapsto \alpha\}\) and \(\iden (e') = \iden_m (e')\) if \(e' \neq e\), and \(\iden (e) = i\) otherwise.
\end{definition}

The memories of any RCCS process could be encoded into \(\iconf\)-structures, but we will encode only reachable (and thus coherent) processes (\autoref{def:coh}).

\begin{figure}
	\begin{minipage}[b]{.09\linewidth}
		\begin{tikzpicture}[scale=0.75]
			\node (emptyset) at (0,0) {\(\emptyset\)};
			\node (a1) at (0, 1.5) {\(\{a\}\)};
			\node (a2) at (0, 3) {\(\{a, b\}\)};
			\draw [->] (emptyset) -- (a1);
			\draw [->] (a1) -- (a2);
		\end{tikzpicture}
		\subcaption{}\label{fig:ex_mem1}
	\end{minipage}
	\hfill
	\begin{minipage}[b]{.09\linewidth}
		\begin{tikzpicture}[scale=0.75]
			\node (emptyset) at (0,0) {\(\emptyset\)};
			\node (a1) at (0, 1.5) {\(\{c\}\)};
			\node (a2) at (0, 3) {\(\{c, \out{a}\}\)};
			\draw [->] (emptyset) -- (a1);
			\draw [->] (a1) -- (a2);
		\end{tikzpicture}
		\subcaption{}\label{fig:ex_mem2}
	\end{minipage}
	\hfill
	\begin{minipage}[b]{.25\linewidth}
		\begin{tikzpicture}[scale=0.75]
			\node (emptyset) at (0,0) {\(\emptyset\)};
			\node (a1) at (0, 1.5) {\(\{(\star, c)\}\)};
			\node (a2) at (0, 3) {\(\{(\star, c), (a, \out{a}) \}\)};
			\node (a3) at (0, 4.5) {\(\{(\star, c), (a, \out{a}), (b, \star) \}\)};
			\draw [->] (emptyset) -- (a1);
			\draw [->] (a1) -- (a2);
			\draw [->] (a2) -- (a3);
		\end{tikzpicture}
		\subcaption{}\label{fig:ex_mem3}
	\end{minipage}
	\hfill
	\begin{minipage}[b]{.22\linewidth}
		\begin{tikzpicture}[scale=0.75]
			\node (emptyset) at (0,0) {\(\emptyset\)};
			\node (a1) at (-1, 1.5) {\(\{a\}\)};
			\node (a2) at (-1, 3) {\(\{a, b\}\)};
			\node (a4) at (1, 3) {\(\{c, a\}\)};
			\node (a3) at (1, 1.5) {\(\{c\}\)};
			\node (a5) at (0, 4.5) {\(\{c, a, b\}\)};
			\draw [->] (emptyset) -- (a1);
			\draw [->] (emptyset) -- (a3);
			\draw [->] (a1) -- (a2);
			\draw [->] (a1) -- (a4);
			\draw [->] (a3) -- (a4);
			\draw [->] (a4) -- (a5);
			\draw [->] (a2) -- (a5);
		\end{tikzpicture}
		\subcaption{}\label{fig:ex_mem4}
	\end{minipage}
	\caption{(Identified) Configurations Structures of Examples~\ref{ex:proc} and \ref{ex:proc2}}
\end{figure}
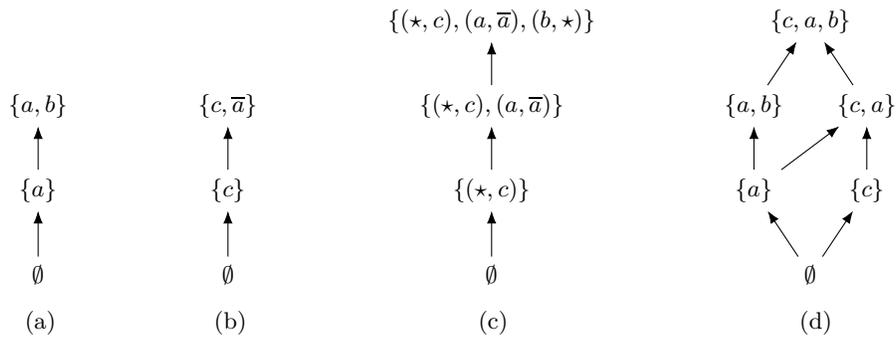

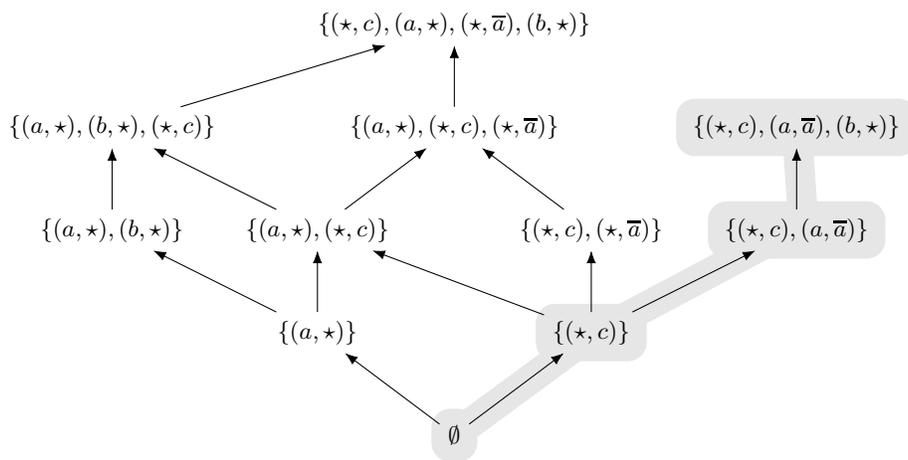
\begin{figure}
	\begin{tikzpicture}[scale=0.91]%

		\node (emptyset) at (0,0) {\(\emptyset\)};
		\node[fill=gray!20, fit=(emptyset), rounded corners=.25cm] {};
		\node (a2) at (2, 1.5) {\(\{(\star, c)\}\)};
		\node[fill=gray!20, fit=(a2), rounded corners=.25cm] {};
		\node (b4) at (5, 3) {\(\{(\star, c), (a, \out{a})\}\)};
		\node[fill=gray!20, fit=(b4), rounded corners=.25cm] {};
		\node (c3) at (5, 4.5) {\(\{(\star, c), (a, \out{a}), (b, \star)\}\)};
		\node[fill=gray!20, fit=(c3), rounded corners=.25cm] {};
		\draw[gray!20, line width=4mm%
			, rounded corners
		] (emptyset) -- ($(a2)+(-0.1, -0.1)$) -- ($(b4)+(0.1, 0.1)$) -- (c3);

		\node (emptyset) at (0,0) {\(\emptyset\)};
		\node (a1) at (-2, 1.5) {\(\{(a, \star)\}\)};
		\node (a2) at (2, 1.5) {\(\{(\star, c)\}\)};
		\node (b1) at (-5, 3) {\(\{(a, \star), (b, \star)\}\)};
		\node (b2) at (-2, 3) {\(\{(a, \star), (\star, c)\}\)};
		\node (b3) at (2, 3) {\(\{(\star, c), (\star, \out{a})\}\)};
		\node (b4) at (5, 3) {\(\{(\star, c), (a, \out{a})\}\)};
		\node (c1) at (-5, 4.5) {\(\{(a, \star), (b, \star), (\star, c)\}\)};
		\node (c2) at (0, 4.5) {\(\{(a, \star), (\star, c), (\star, \out{a})\}\)};
		\node (c3) at (5, 4.5) {\(\{(\star, c), (a, \out{a}), (b, \star)\}\)};
		\node (d1) at (0, 6) {\(\{(\star, c), (a, \star), (\star, \out{a}), (b, \star)\}\)};
		\draw [->] (emptyset) -- (a1);
		\draw [->] (emptyset) -- (a2);
		\draw [->] (a1) -- (b1);
		\draw [->] (a1) -- (b2);
		\draw [->] (a2) -- (b2);
		\draw [->] (a2) -- (b3);
		\draw [->] (a2) -- (b4);
		\draw [->] (b1) -- (c1);
		\draw [->] (b2) -- (c1);
		\draw [->] (b2) -- (c2);
		\draw [->] (b3) -- (c2);
		\draw [->] (b4) -- (c3);
		\draw [->] (c1) -- (d1);
		\draw [->] (c2) -- (d1);
	\end{tikzpicture}
	\caption{Configurations Structure of \autoref{ex:proc} (Continued)}
	\label{fig:ex_mem5}
\end{figure}

\begin{example}
	\label{ex:proc}
	Consider the following transitions:
	\begin{align*}
		\emptymem \rhd (a.b\mid c.\out{a}) & \congru (\fork . \emptymem \rhd a.b)\mid(\fork . \emptymem \rhd c.\out{a})                                                                  \\
		                                   & \fwlts{1}{c} (\fork . \emptymem \rhd a.b)\mid(\mem{1,c,0}. \fork . \emptymem \rhd \out{a})                                                  \\
		                                   & \fwlts{2}{\tau} (\mem{2, a, 0} . \fork . \emptymem \rhd b) \mid (\mem{2, \out{a}, 0} . \mem{1,c,0} . \fork .\emptymem \rhd 0)               \\
		                                   & \fwlts{3}{b} (\mem{3, b, 0} . \mem{2, a, 0} . \fork . \emptymem \rhd 0) \mid (\mem{2, \out{a}, 0} . \mem{1,c,0} . \fork . \emptymem \rhd 0)
	\end{align*}

	Now, observe that:
	\begin{align*}
		  & \encm{(\mem{3, b, 0} . \mem{2, a, 0} . \fork . \emptymem \rhd 0) \mid (\mem{2, \out{a}, 0} . \mem{1,c,0} . \fork . \emptymem \rhd 0)}    \\
		= & \encm{\mem{3, b, 0} . \mem{2, a, 0} . \fork . \emptymem \rhd 0} \mid \encm{\mem{2, \out{a}, 0} . \mem{1,c,0} . \fork . \emptymem \rhd 0} \\
		= & \encm{\mem{3, b, 0} . \mem{2, a, 0} . \fork . \emptymem} \mid \encm{\mem{2, \out{a}, 0} . \mem{1,c,0} . \fork . \emptymem}
	\end{align*}
	If we let (with the obvious labeling functions \(\labl(x) = x\) for \(x \in \{a, \out{a}, b, c\}\)):
	\begin{align*}
		\encm{\mem{3, b, 0}.\mem{2, a, 0}.\fork.\emptymem}       & = (\{a, b\}, \{ \emptyset, \{a\}, \{b,a\}\}, \{a \mapsto 2, b \mapsto 3
		\}                                                                                                                                                    \\
		\encm{\mem{2, \out{a}, 0}.\mem{1, c, 0}.\fork.\emptymem} & = (\{c, \out{a}\}, \{ \emptyset, \{c\}, \{c,\out{a}\} \}, \{c \mapsto 1, \out{a} \mapsto 2
		\})
	\end{align*}
	Then we respectively get the structures of Figures~\ref{fig:ex_mem1} and~\ref{fig:ex_mem2}.

	The product of those two configurations results in the following set of events (where we keep naming the events after their label):

	\begin{multicols}{4}
		\begin{subequations}
			\renewcommand{\theequation}{ev.\ \arabic{equation}}
			\noindent \begin{align}
				(a, \star) \label{bad4} \\
				(b, \star)
			\end{align}
			\begin{align}
				(\star, c) \\
				(\star, \out{a}) \label{bad5}
			\end{align}
			\begin{align}
				(a, c) \label{bad1} \\
				(a, \out{a})
			\end{align}
			\begin{align}
				(b, c) \label{bad2} \\
				(b, \out{a}) \label{bad3}
			\end{align}
		\end{subequations}
	\end{multicols}
	When doing the parallel composition, the relabeling labels with \(\bot\) events \ref{bad1}, \ref{bad2} and \ref{bad3}, since their event identifier do not match, as well as \ref{bad4} and \ref{bad5}, since they have the same identifiers %
	but occur unsynchronized.
	Hence, we obtain the \(\iconf\)-structure of \autoref{fig:ex_mem3},
	with \(\iden(\star, c) = 1\), \(\iden(a, \out{a}) = 2\) and \(\iden(b, \star) = 3\).

	Observe the encoding of the origin of this process, \(a.b\mid c.\out{a}\), in \autoref{fig:ex_mem5}: the underlying configuration of the encoding of the memory previously exposed is just a particular path in that configuration structure.
	This remark is made formal in \autoref{lem:mem_enc_in_origin}.
\end{example}

\begin{example}
	\label{ex:proc2}
	Similarly, we can encode the execution
	\[\emptymem\rhd (a.b\mid c) \fwlts{1}{c} \fwlts{2}{a} \fwlts{3}{b} (\mem{3, b, 0} . \mem{2, a, 0} . \fork . \emptymem \rhd 0) \mid (\mem{1,c,0} . \fork . \emptymem \rhd 0)\]
	and obtain the structure of \autoref{fig:ex_mem5}, with \(\iden(c) = 1, \iden(a) = 2\) and \(\iden(b) = 3\).
\end{example}

For the rest of this subsection, we assume given a coherent reversible process \(R\) whose origin is written \(\orig{R}\). %
As expected the underlying configuration of %
\(\encm{R}\) %
is included in \(\enc{\orig{R}}\) (this will be a direct consequence of \autoref{lem:enc_mem_2}): %

\begin{lemma}
	\label{lem:mem_enc_in_origin}
	\(\functoric(\encm{R})\subseteq \enc{\orig{R}}\)\footnote{Formally, \(\conf_1 = (E_1,C_1,\labl_1)\subseteq\conf_2= (E_2,C_2,\labl_2)\) iff \(E_1\subseteq E_2\), \(C_1\subseteq C_2\) and \(\labl_1\subseteq \labl_2\).}.
\end{lemma}

The following lemma states that all memory events in \(R\) %
are either causally linked or concurrent.
Therefore their encoding results in partially ordered set (poset) with one maximal element (\autoref{def:causality}), linked by the subset relation.
This is illustrated by %
Examples~\ref{ex:proc} and \ref{ex:proc2}, and proved by induction on the RCCS process.

\begin{lemma}
	\label{lem:prefix_enc}
	\(\encm{R}\) is a poset with one maximal configuration.%
\end{lemma}

Another way to encode %
a reachable memory in a configuration structure~\cite{Aubert2016jlamp} is to encode %
\(R\) %
as a configuration, called the \emph{address} of \(R\), inside %
\(\enc{\orig{R}}\)\footnote{By abuse of notation, \(\enc{\cdot}\) denotes two different encodings, on CCS and RCCS processes.%
}:
\begin{align*}
	\enc{R} = (\enc{\orig{R}}, x) \text{ where }x \text{ is a configuration in }\enc{\orig{R}}\text{.}
\end{align*}
\begin{definition}[Generation]
	\label{def:xdownarrow}
	Given \((E,C,\labl)\) and \(x \in C\), %
	the \emph{configuration structure generated by \(x\)} is \(x\downarrow = (x, C_x, \labl \restr{x})\), where \(C_x = \{ y \setst y\in C, y\subseteq x\}\)%
	\footnote{For the reader familiar with event structures, a configuration \(x\) defines an event structure \((x, \leqslant_x, \labl)\) for \(\leqslant_x\) the causality in \(x\).
		The construction here mirrors the transformation from an event structures to a configuration structure~\cite{Winskel1995}.}.
\end{definition}

The encoding of the memory of \(R\), \(\encm{R}\), is the configuration structure generated by the configuration \(x\).
In other words, the memory of \(R\) is \enquote{everything under} the address of \(R\) in \(\enc{\orig{R}}\).
We state this formally in the following lemma:

\begin{lemma}
	\label{lem:enc_mem_2}
	For %
	\(\enc{R} = (\enc{\orig{R}}, x)\), %
	\(\functoric(\encm{R}) \iso x \downarrow\).
\end{lemma}

We can now make the intuitions of \autoref{fig:schema} formal, letting \(R\equiv (\prod_{i} m_i\rhd P_i)\bs \namelist{a}\) and \(\orig{R}\) be its origin.
Informally speaking, the memories of \(R\), i.e. the \(m_i\), %
is the \enquote{past} of \(R\), and the process \((\prod_{i} P_i)\bs \namelist{a}\) is its \enquote{future}: %
in the configuration structure \(\enc{\orig{R}}\), they are respectively represented by \(\functoric(\encm{R})\) and \(\enc{(\prod_{i} P_i)\bs \namelist{a}}\).
Consider \autoref{ex:proc}: the past of the process is the structure shown in \autoref{fig:ex_mem3}, which corresponds to the gray part in \autoref{fig:ex_mem4}.
The future of this process is empty, as all actions have been consumed.

However, we cannot recover the origin process from the encoding of the past and future of a reversible process.
There is a loss of information that occurs for the synchronisation events when encoding memories into identified configuration structures: we label \(\tau\) such events, whereas in the memories of a RCCS process the synchronised events keep their input (or output) label.

\subsection{(Hereditary) History-Preserving Bisimulations on CCS}
\label{sec:equiv_rccs}
We adapt the (hereditary) history-preserving bisimulations of \autoref{sec:hpbs} %
to CCS, making %
the bijections %
becomes isomorphisms between memory encodings.

Below, %
we let \(P_1\) and \(P_2\) be two CCS processes (possibly with auto-concurrency, \autoref{def:autocrccs}), and \(\enc{P_i} = (E_i,C_i,\labl_i)\), for \(i = 1, 2\).

\begin{definition}[\HPB and \HHPB on CCS]
	\label{def:hpb_rccs}
	A relation \(\rel\subseteq \rproc \times \rproc \times (E_1 \rightharpoonup E_2)\) such that \((\emptymem\rhd P_1,\emptymem\rhd P_2,\emptymem)\in\rel\) and
	if \((R_1, R_2, f) \in \rel\) then \(f\) is an isomorphism between \(\functoric(\encm{R_1})\) and \(\functoric(\encm{R_2})\) and (\ref{hpb_rccs1}) and (\ref{hpb_rccs2}) (resp.\ (\ref{hpb_rccs1}--\ref{hpb_rccs4})) hold is called a \emph{history-}(resp.\ \emph{hereditary history-}) \emph{preserving bisimulation between \(P_1\) and \(P_2\)}.
	\begin{align}
		\forall S_1, R_1\fwlts{i}{\alpha}S_1 \Rightarrow
		\exists S_2, g, R_2\fwlts{j}{\alpha}S_2, g\restr{\functoric(\encm{S_1})} = f, (S_1, S_2, g)\in\rel \label{hpb_rccs1}  \\
		\forall S_2, R_2\fwlts{i}{\alpha}S_2 \Rightarrow
		\exists S_1, g, R_1\fwlts{j}{\alpha}S_1, g\restr{\functoric(\encm{S_1})} = f , (S_1, S_2, g)\in\rel \label{hpb_rccs2} \\
		\forall S_1, R_1\bwlts{i}{\alpha}S_1 \Rightarrow
		\exists S_2, f, R_2\bwlts{j}{\alpha}S_2, g = f\restr{\functoric(\encm{S_1})}, (S_1, S_2, g)\in\rel \label{hpb_rccs3}  \\
		\forall S_2, R_2\bwlts{i}{\alpha}S_2 \Rightarrow
		\exists S_1, g, R_1\bwlts{j}{\alpha}S_1, g = f\restr{\functoric(\encm{S_1})}, (S_1, S_2, g)\in\rel \label{hpb_rccs4}
	\end{align}
\end{definition}

\begin{definition}[\wfHPB on CCS]
	\label{def:whpb_rccs}
	A \emph{weak-function history-preserving bisimulation between \(P_1\) and \(P_2\)} is a relation \(\rel\subseteq \rproc \times \rproc \times (E_1 \rightharpoonup E_2)\) such that
	\((\emptymem\rhd P_1,\emptymem\rhd P_2,\emptymem) \in \rel\) and
	if \((R_1, R_2,f) \in \rel\), then \(f\) is a isomorphism between \(\functoric(\encm{R_1})\) and \(\functoric(\encm{R_2})\) and

	\begin{align*}
		\forall S_1, R_1\fwlts{i}{\alpha}S_1 \Rightarrow
		\exists S_2, g, R_2\fwlts{j}{\alpha}S_2, (S_1,S_2, g)\in\rel \\
		\forall S_2, R_2\fwlts{i}{\alpha}S_2 \Rightarrow
		\exists S_1, g, R_1\fwlts{j}{\alpha}S_1, (S_1, S_2, g)\in\rel
	\end{align*}
\end{definition}

Of course, \wfHHPB on CCS is defined similarly.

Note that the definitions above reflect definitions \ref{def:hpb} and \ref{def:whpb}: the condition \((\emptymem\rhd P_1,\emptymem\rhd P_2,\emptymem) \in \rel\) is intuitively the counterpart to the condition that \((\emptyset,\emptyset,\emptyset)\) %
have to be included in the relation on configuration structures. %
Also, \(f\) shares similarity with the label- and order-preserving bijection. %
Finally, note that, to keep the definitions concise, we wrote \(\rproc\), but one needs only to consider the reachable processes from \(\emptymem\rhd P_1\) and \(\emptymem\rhd P_2\).

\begin{theorem}[Main result]
	\label{th:hpb_rccs}
	\(P_1\) and \(P_2\) are \HHPB (resp.\ \wfHHPB, \HPB, \wfHPB) iff \(\enc{P_1}\) and \(\enc{P_2}\) are.
\end{theorem}

\section{Conclusion}
\label{sec:related_work}
\label{sec:conclusion}
In this paper, we recalled how the previous attempt to characterize syntactically \HHPB gave partial result (\autoref{thm:previousresult}).
Then, we defined a series of bisimulations on CCS processes that corresponds to \HPB, \HHPB, and their \enquote{weak function} variants, on configuration structures. We managed therefore to define an equivalence on CCS which distinguishes for instance \(a.a \mid b\) from \(a \mid a \mid b\), whose encodings are not \HHPB.
We would like to conclude with three observations.

Fist, we should stress that our relations are defined in terms of CCS processes: on the surface, this paper offers a new result on \emph{non-reversible} CCS, using tools stemming from the study of reversible computation.
We believe this is an interesting contribution, that witnesses the relevance of studying concurrent reversible computation.

Our paper also introduces a natural technical tool, identified configuration structure, to encode information on memory events.
It should be noted that a memory event is made of \emph{three} elements, \(\mem{i, \alpha, Q}\), and that configuration structures encode names \(\alpha\), and \(\iconf\)-structures furthermore encode identifiers \(i\).
We could take a step further and also enclose information about the residual process \(Q\) in a sum.
Representing the \enquote{whole memory} in denotational model could lead to interesting new bisimulation on processes, but we leave this for a future possible extension.

Other CCS bisimulations such as the \emph{pomset} bisimulations~\cite{Boudol1987} or the \emph{localities} bisimulations~\cite{Boudol1992} are known to be different from the history-preserving ones.
However, these bisimulations add some information on auto-concurrent events, that can be used to distinguish them.
A possible direction of future work is then to adapt these bisimulations to the reversible setting to maybe capture \HHPB.

\bibliographystyle{splncs04}
\renewcommand{\doi}[1]{\url{https://doi.org/#1}}

\clearpage

\appendix

\section{Appendix}
\label{sec:appendix}

The appendix is divided in three subsections.

\autoref{subsec:apA} gathers results about identified configuration structures, introduced in \autoref{sec:ics}.
They are not required to understand the rest of the technical development, but illustrates why we believe this notion is \enquote{solid}, and gives some insights on how to manipulate it.

The main purpose of \autoref{subsec:apB} is to prove \autoref{lem:enc_mem_2}, stated in \autoref{sec:enc_mem}.
This apparently intuitive result actually requires a heavy machinery to be proven: not only do we prove \autoref{lem:prefix_enc}, but we also state and prove some intermediate lemmas.

The immediate advantage of \autoref{lem:enc_mem_2} is that it makes the proof of \autoref{th:hpb_rccs}, in \autoref{subsec:apC}, almost straightforward.

Before doing so, we need to prove the only new result of \autoref{sec:prelim}:

\begin{proof}[\autoref{cor:unique_id}, page~\pageref{cor:unique_id}]
	Since \(m \rhd P\) is reachable, \(m \rhd P\) is coherent by \autoref{lem:reach-coh}, and hence \(m \coh \emptyset\) by \autoref{def:coh}.
	Looking at \autoref{def:coherence}, the only way to derive \(m \coh \emptyset\) is to start with em., and then to apply ev.\ to \enquote{stack} the memory events.
	If an identifier were to appear in two memory event in \(m\), we would not be able to use ev.\ to add the second memory event, since the side-condition would forbid it.
	\qed\end{proof}

\subsection{On the Robustness of Identified Configuration Structures}%
\label{subsec:apA}

This subsection gathers elementary results on the structures introduced in \autoref{sec:ics}.
These results, e.g., that the category of \(\iconf\)-structures of \autoref{def:cat_ics} is indeed a category (\autoref{lem:cat_ics}) or that the operations on \(\iconf\)-structures of \autoref{icat-op-def} are well-defined (\autoref{lem:preserve_ics}), were not stated in the body of the document, but are implicitely used.

\begin{lemma}
	\label{lem:cat_ics}
	Identified configuration structure and their morphisms (\autoref{def:cat_ics}, page~\pageref{def:cat_ics}) form a category .
\end{lemma}

\begin{proof}
	\begin{description}
		\item[Identity] For every \(\iconf\)-structure \(\iconf = (E, C, \labl, \ids, \iden)\), \(\id_{\iconf} : \iconf \to \iconf\) is defined to be the identity on the underlying configuration structure \(\id : (E, C, \labl) \to (E, C, \labl)\) from \(\catcs\), that trivially preserves identifiers.
		      For any morphism \(f : \iconf_1 \to \iconf_2\), \(f \circ \id_{\iconf_1} = f = \id_{\iconf_2}\circ f\) is trivial.
		\item[Associativity] for \(f : \iconf_1 \to \iconf_2\), \(g : \iconf_2 \to \iconf_3\) and \(h : \iconf_3 \to \iconf_4\), \(h \circ (g \circ f) = (h \circ g) \circ f\) is inherited from the associativity in \(\catcs\), and since \(f\), \(g\) and \(h\) all preserves identifiers.
	\end{description}
	Hence \(\catics\) is a category.
	\qed
\end{proof}

\begin{lemma}
	\label{lem:product_ics}
	The product of \(\iconf\)-structures of \autoref{icat-op-def}, page~\pageref{icat-op-def}, is the product in \(\catics\).%
\end{lemma}

\begin{proof}%
	First note that \(\iconf_1\times\iconf_2\) is a \(\iconf\)-structures as it is a configuration structure and from the definition of \(\iden\) every event in a configuration has a unique label.
	Secondly, it is easy to show that the projections are morphisms. Lastly to show that the structure \(\iconf_1\times\iconf_2\) has the universal property, we proceed in two steps:
	\begin{itemize}
		\item the underlying configuration structure is the product of the underlying configuration structures, by definition:
		      \begin{align*}
			      \functoric(\iconf_1\times\iconf_2) = \functoric(\iconf_1)\times\functoric(\iconf_2);
		      \end{align*}
		\item for any \(\iconf'\) which projects into \(\iconf_1\) and \(\iconf_2\), then \(\functoric(\iconf')\) projects into \(\functoric(\iconf_1)\) and \(\functoric(\iconf_2)\) and therefore there exists a unique morphism \(h:\functoric(\iconf')\to \functoric(\iconf_1\times\iconf_2)\). It is easy to show that since the projections preserve identifiers, then so does \(h\) which concludes our proof.
	\end{itemize}
	This lemma also follows from \cite[Proposition 85]{Winskel1995}.
	\qed
\end{proof}

\begin{lemma}
	\label{lem:preserve_ics}
	The operations of \autoref{icat-op-def} (product, relabeling, restriction and parallel composition), page~\pageref{icat-op-def}, preserve \(\iconf\)-structures.
\end{lemma}

\begin{proof}%
	Let us note that (i) the product, relabeling, restriction, and parallel composition on configuration structures from \autoref{cat-op-def} preserve configuration structures and that (ii) any configuration structure endowed with a valid identifier function (i.e., such that no two events in the same configuration have the same identifier, cf. \ref{eq:collision}) is a valid \(\iconf\)-structure.

	For the product, it follows trivially from \autoref{lem:product_ics}.

	Relabeling does not change anything but the labels, so we have nothing to prove.

	The restriction only removes events in configurations and keeps the identifier function intact. Hence if the initial structure has a valid identifier function, then the identifier function of the new structure is a valid one by assumption.

	Let us now consider the parallel composition of two \(\iconf\)-structures, denoted \(\iconf_1 \mid \iconf_2 = (E, C, \labl, \iden)\). Proving that the identifier function is valid follows from a case analysis. %
	Given a configuration \(x \in C\) and two events \(e,e' \in x\), these are the possible cases:%
	\begin{itemize}
		\item \(\pi_2(e)= \star\) and \(\pi_2(e')= \star\).
		      In this case, looking at the definition of the product in \(\iconf\)-structures, \(\iden(e) = (\iden_1(\pi_1(e)), \iden_{\star})\) and \(\iden(e') = (\iden_1(\pi_1(e')), \iden_{\star})\).
		      If \(\iden(e) = \iden(e')\), then \(\iden_1(\pi_1(e)) = \iden_1(\pi_1(e'))\) in the configuration \(\pi_1(x)\) in \(\iconf_1\).
		      But that's a contradiction, since \(\pi_1(e)\) and \(\pi_1(e')\) are in the same configuration and the identifier function of \(\iconf_1\) is valid.
		\item \(\pi_1(e)= \star\) and \(\pi_1(e')= \star\). This case is similar as the previous one, except that it uses that the identifier function of \(\iconf_2\) is valid.
		\item \(\pi_1(e)\neq \star\) and \(\pi_2(e')\neq \star\) (with either \(\pi_1(e') = \star\) or \(\pi_1(e')\neq \star\)).
		      If \(\iden(e) = \iden(e')\), then \(\iden_i(\pi_i(e)) = \iden_i(\pi_i(e'))\) for \(i = 1, 2\).
		      Then in this case,
		      \begin{itemize}
			      \item either one of them, say, \(e\), is a synchronisation or a fork: in this case, \(\iden_1(\pi_1(e)) = \iden_2(\pi_2(e)) = \iden_2(\pi_2(e'))\), and \(e'\) was relabeled \(\bot\) at the relabeling stage of the parallel composition, and then removed during the restriction. Hence a contradiction: \(e\) and \(e'\) can't be two events in the same configuration.
			      \item or none of them is a synchronisation, in which case both events were removed by the restriction.
			            Hence, again, a contradiction: \(e\) and \(e'\) can't be two events in the same configuration.
		      \end{itemize}
		      The symmetric case (where \(\pi_2(e)\neq\star\) and \(\pi_1(e')\neq\star\)) is similar.
	\end{itemize}
	\qed
\end{proof}

The following lemma makes more formal the intuition of \autoref{remark1}, page~\pageref{remark1}.
Remember that we assumed that for every configuration structure \((E, C, \labl)\), \(E\) was endowed with a total order, that we write \(\preceq\).

\begin{lemma}
	\label{lem:functors}
	\(\functoric : \catics \to \catcs\), defined by
	\begin{itemize}
		\item \(\functoric (E, C, \labl, \ids, \iden) = (E, C, \labl)\)
		\item %
		      \(\functoric(f_E, f_C, f_{\iden}) = (f_E, f_C)\)
	\end{itemize}
	and \(\functorci : \catcs \to \catics\), defined by
	\begin{itemize}
		\item \(\functorci (E, C, \labl) = (E, C, \labl, \ids, \iden)\), where \(\ids = \{1, \hdots, \card{E}\}\) and
		      \[
			      \iden(e) = \begin{cases*}
				      1 & if \(\forall e', e \preceq e'\) \\
				      i+1 & if \(\exists e', e' \preceq e\), \(\iden(e') = i\) and there is no \(e''\) \st \(e' \preceq e'' \preceq e\)
			      \end{cases*}
		      \]
		\item For \((f_E, f_C): (E_1, C_1, \labl_1) \to (E_2, C_2, \labl_2)\), \(\functorci(f_E, f_C) = (f_E, f_C, f_{\iden})\), where \(f_{\iden}(\iden_1(e)) = \iden_2(f_E(e_2))\).
	\end{itemize}
	are functors.
\end{lemma}
\begin{proof}
	Proving that \(\functoric\) is a functor is immediate.

	Proving that \(\functorci(\conf)\) is a \(\iconf\)-structure is immediate, since our construction of \(\iden\) trivially insures \ref{eq:collision}.
	For \((f_E, f_C): \conf_1 \to \conf_2\), proving that \(\functorci(f_E, f_C)\) is a morphism between \(\functorci(\conf_1)\) and \(\functorci(\conf_2)\) is also immediate.
	For the preservation of the identity, we compute:
	\begin{align*}
		\functorci(\id_{\conf}) & = \functorci(\id_E, \id_C)   \\
		                        & = (\id_E, \id_C, f_{\iden})  \\
		\shortintertext{where \(f_{\iden} (\iden (e)) = \iden(\id_E(e)) = \iden(e)\), hence \(f_{\iden} = \id_{\ids} : \ids \to \ids\),}
		                        & = (\id_E, \id_C, \id_{\ids}) \\
		                        & = \id_{\functorci(\conf)}
	\end{align*}
	For the composition of morphisms, given \(f = (f_C, f_E) : \conf_1 \to \conf_2\) and \(g = (g_C, g_E) : \conf_2 \to \conf_3\), we write \(\functorci(\conf_i) = (E_i, C_i, \labl_i, \ids_i, \iden_i)\) and we compute:
	\begin{align*}
		\functorci(g) \circ \functorci(f) & = (g_C, g_E, g_{\iden}) \circ (f_C, f_E, f_{\iden})         \\
		                                  & = (g_C \circ f_C, g_E \circ f_E, g_{\iden} \circ f_{\iden}) \\
		\shortintertext{where, for all \(e \in E_1\), we compute:
			\begin{align*}
				(g_{\iden} \circ f_{\iden})(\iden_1(e)) & = g_{\iden} (f_{\iden}(\iden_1(e)) \\
				                                        & = g_{\iden} (\iden_2(f_E(e)))      \\
				                                        & = \iden_3(g_E (f_E(e)))
			\end{align*}
			And hence we can conclude:
		}
		\functorci(g) \circ \functorci(f) & = \functorci(g \circ f)
	\end{align*}
	\qed
\end{proof}

\subsection{Proofs for \autoref{sec:enc_mem}}
\label{subsec:apB}
In the following, we start by observing that \autoref{lem:mem_enc_in_origin} follows from \autoref{lem:enc_mem_2}.
\autoref{lem:enc_mem_2}, on its side, requires a bit of work: on top of proving \autoref{lem:prefix_enc}, we state and prove some intermediate lemmas (Lemma~\ref{lem:id_unique}, \ref{lem:causality_mem} and \ref{lem:te}) needed to obtain it.

\begin{proof}[\autoref{lem:mem_enc_in_origin}, page~\pageref{lem:mem_enc_in_origin}]
	Follows from \autoref{lem:enc_mem_2}. %
	\qed
\end{proof}

\begin{lemma}
	\label{lem:id_unique}
	For all RCCS process \(R\), letting \(\encm{R} = (E, C, \iden)\),
	for all \(e_1, e_2 \in E\), \(\iden(e_1) = \iden(e_2)\) implies \(e_1 = e_2\).
\end{lemma}
\begin{proof}
	We proceed by structural induction on \(R\).
	From \autoref{cor:unique_id} the only interesting case is the parallel composition, i.e. \(R= R_1\mid R_2\).
	From the definition of parallel composition in \(\iconf\)-structures (\autoref{icat-op-def}), it follows that \(\iden(e_1) = \iden(e_2)\) implies \(e_1 = e_2\).
	\qed
\end{proof}

\begin{proof}[\autoref{lem:prefix_enc}, page~\pageref{lem:prefix_enc}]
	We proceed by induction on \(R\).

	If \(R\) is \(m \rhd P\), we prove that \(\encm{m}\) is a poset with one maximal element by induction on \(m\).
	The base case, if \(m\) is \(\emptymem\), is trivial, since \(\encm{\emptymem} = \confzero\) is a poset with one maximal element.
	If \(m\) is \(\fork . m'\), then it follows by induction hypothesis, since \(\encm{\fork . m'} = \encm{m'}\).
	If \(m\) is \(\mem{i, \alpha, P} . m'\), then by induction hypothesis, \(\encm{m'}\) is a poset with one maximal element, and the construction of \(\encm{\mem{i, \alpha, P} . m'}\) detailed in \autoref{def:encode_m} preserves that property.

	If \(R\) is \(R' \bs a\), then it trivially follows by induction hypothesis.

	Finally, if \(R\) is \(R_1 \mid R_2\), then by induction hypothesis we get that \(\encm{R_1}\) and \(\encm{R_2}\) are both posets with a maximal configuration.
	We also know by \autoref{lem:id_unique} that they have disjoint identifiers.
	Looking at the definition of parallel composition for \(\iconf\)-structures (\autoref{icat-op-def}), we may observe that \(\encm{R} = \encm{R_1 \mid R_2} = \encm{R_1} \mid \encm{R_2} = \big(\relabl \circ (\encm{R_1}\times\encm{R_2})\big)\restr{F}\).

	We show that there exists more than one maximal configurations in \(\encm{R_1}\times\encm{R_2}\) and that all but one are removed by the restriction.

	We show this by first showing that there exists more than one maximal configurations in \(\functoric(\encm{R_1})\times\functoric(\encm{R_2})\), denoted here with \(\conf\).
	From the definition of product (\autoref{cat-op-def}), we have that there exists \(y_1,\cdots y_n\) maximal configurations in \(\conf\) such that \(\pi_1(y_i)\) and \(\pi_2(y_i)\) are maximal in \(\conf_1\) and \(\conf_2\), respectively.
	As \(\encm{R_1}\times\encm{R_2} = (\functoric(\encm{R_1})\times\functoric(\encm{R_2})) + \iden\) it follows that the maximal configurations of \(\conf\) are preserved in \(\encm{R_1}\times\encm{R_2}\).

	A second step is then to show that the restriction keeps only one maximal configuration.
	Let \(y_i,y_j\) be two maximal configurations. As they are maximal it implies that \(y_i\cup y_j\notin\conf\), for \(i\neq j\leqslant n\). In turn, this implies that there exists \(e_i\in y_i\) and \(e_j\in y_j\) such that \(\pi_1(e_i) = \pi_1(e_j)\) and \(e_i\neq e_j\), as otherwise \(y_i\cup y_j\) would be defined. Here we assume that \(\pi_1(e_i) = \pi_1(e_j)\) but we could also take \(\pi_2(e_i) = \pi_2(e_j)\) and the argument still holds.

	Let us now take \(d\) an event in \(\conf_1\) and take \(e_1, \hdots, e_m\) the subset of events in \(E\) where \(\pi_1(e_i) = d\). The restriction in the parallel composition of \(\encm{R_1} \mid \encm{R_2}\) keeps only one such event \(e_i\) and removes the rest. Therefore, from all maximal configurations \(y_1, \hdots, y_m\) such that \(e_i \in y_i\), \(i \leq m\), only one remains.

	By applying the argument above to all events in \(\encm{R_1}\) (and \(\encm{R_2}\)), we have that the restriction removes all but one \(y_i\), which is then the maximal configuration in \(\encm{R_1} \mid \encm{R_2}\).

	\qed
\end{proof}

Let us write \(\xmax\) the maximal configuration from \autoref{lem:prefix_enc}.

For the following proof, we need to introduce the causality relation on memory events and on transitions from \cite{Danos2005}. We write \(i_1<_R i_2\) for \(R\) a proces and \(i_1,i_2\in \ids\) if there exist a memory stack \(m.\mem{i_1,\alpha_1,P_1}.m'.\mem{i_2,\alpha_2,P_2}.m''\) in \(R\). We write \(i_1\leq_R i_2\) for the transitive closure of \(<\) over all memories of \(R\).

\begin{lemma}
	\label{lem:causality_mem}
	Let \(R\) be a reversible process and \(\encm{R} = (E,C,\labl,\iden)\) be the encoding of its memory, with \(\xmax\) the maximal configuration in C. Let \(\mem{i_1,\alpha_1,P_1}\) and \(\mem{i_2,\alpha_2,P_2}\), \(i_1\neq i_2\) be two memory events in \(R\) and let \(e_1,e_2\in E\) be the two corresponding events, i.e. \(\iden(e_j) = i_j\), \(j\in\{1,2\}\). Then \(i_1\leq_R i_2 \iff e_1 \leq_{\xmax} e_2\).
\end{lemma}

\begin{proof}
	Follows by a structural induction on \(R\).
\end{proof}

We also import from \cite[Definition 1]{Danos2005} the definition of causality on transitions \(t_1\) and \(t_2\), of a trace \(\theta\), that we denote here with \(t_1\leq_{\theta} t_2\). We do not give the definition formally, as it is lengthy, but intuitively, it is the causality relation on memory events lifted to transitions.

Let \(R\) be an RCCS process and \(\theta:\orig{R}\tfbwlts R\) a trace. Also, let \(\enc{\orig{R}} = (E,C,\labl)\) and let \(x\) be a configuration such that \(\enc{R} = (\enc{\orig{R}}, x)\). We define a bijection from transitions in \(\theta\) to events in \(x\), that we write \(\te{}\). We define the bijection by induction on the trace \(\theta\) using the fact for each transition \(t:R'\redl{i,\alpha} R''\) in \(\theta\), there exists \(x'\) a configuration and \(e'\) an event in \(\enc{\orig{R}}\) such that \(\enc{R'} = (\enc{\orig{R}},x')\) and \(\enc{R''} = (\enc{\orig{R}},x'\cup\{e'\})\). Then \(\te{}(t) = e'\).

\begin{lemma}
	\label{lem:te}
	Let \(t_1\) and \(t_2\) be two transitions of a trace \(\theta:\orig{R}\tfbwlts R\) and let \(\enc{R} = (\enc{\orig{R}}, x)\). Then \(t_1\leq_{\theta} t_2\iff \te{}(t_1)\leq_{x} \te{}(t_2)\).
\end{lemma}

\begin{proof}
	Follows by induction on the trace \(\theta\).
\end{proof}

\begin{proof}[\autoref{lem:enc_mem_2}, page~\pageref{lem:enc_mem_2}]
	We reformulate the hypothesis and show a stronger (in the sense that it is more specific) result from which~\autoref{lem:enc_mem_2} follows.

	As \(R\) is reachable there exists a forward-only trace \(\orig{R}\redl{i_1,\alpha_1}\cdots\redl{i_n,\alpha_n} R\)~\cite[Lemma 1]{Aubert2016jlamp}, denoted by \(\theta\).
	We write \(\encm{R_j}= (E_j,C_j,\labl_j,\iden_j)\) for \(R_j\) a process in the trace above and \(\theta_j :\orig{R}\redl{i_1,\alpha_1}\cdots\redl{i_j,\alpha_j} R_j\) a subtrace of \(\theta\).

	Let \(\enc{\orig{R}} = (E,C,\labl)\). From \cite{Aubert2016jlamp} we have that \(\enc{R_j} = (\enc{\orig{R}}, x_j)\) for \(x_j\in C\).
	We show that there exists three bijections:
	\begin{itemize}
		\item \(\te{j}\) between transitions in the trace \(\theta_j\) and events in \(\enc{\orig{R}}\), i.e. \(\te{j}: \theta_j\to E\);
		\item \(\mt{j}\) between events in \(\encm{R_j}\) and transitions, i.e. \(\mt{j}:E_j\to\theta_j\);
		\item \(\me{j}\) between events in \(\encm{R_j}\) and events in \(\enc{\orig{R}}\), i.e. \(\me{j}:E_j\to E\), such that \(\me{j} = \te{j}\circ\mt{j}\).
	\end{itemize}
	Moreover, all three bijections preserve the labels and the causality relations. In particular, \(\me{}\) is label and order preserving for all events in \(\xmax\), the maximal configuration in \(\encm{R_j}\):
	\begin{align}
		\me{j}(\xmax) = x_j \\
		\label{eq:causality_max}
		e_1\leq_{\xmax} e_2 \iff \me{j}(e_1) \leq_{x_j} \me{j}(e_2).
	\end{align}
	From~\autoref{def:xdownarrow}, \(x\downarrow = (x,C_x,\labl\restr x)\), for \(C_x = \{ y \setst y\in C, y\subseteq x\}\). We can write \(\encm{R_j} = \xmax \downarrow\), and therefore have \(\functoric(\encm{R_j}) \iso (\me{j}(\xmax)) \downarrow\).

	We proceed by induction on the trace \(\theta\), and at each step we extend the three bijections above.

	For a transition \(R_j\redl{i,\alpha} R_{j+1}\) and a trace \(\theta_j\) we have by induction that
	\begin{itemize}
		\item \(\functoric(\encm{R_j}) \iso x_j \downarrow\), for \(\enc{R_j} = (\enc{\orig{R}}, x_j)\);
		\item there exists \(\te{j}:\theta_j\to E\), \(\mt{j}:E_j\to\theta_j\) and \(\me{j}:E_j\to E\)
	\end{itemize}
	as defined above.

	There is an operational correspondence between \(R_j\) and its encoding~\cite[Lemma 6]{Aubert2016jlamp}, and therefore there exists \(e\) an event in \(\enc{\orig{R}}\)
	such that
	\begin{align}
		\label{eq:op_corr}
		(\enc{\orig{R}},x_j)\redl{e}(\enc{\orig{R}},x_j\cup\{e\})
	\end{align}
	with \(\labl(e)=\alpha\) and where we write \(x_{j+1} = x_j\cup\{e\}\).

	We have to show that \(\functoric(\encm{R_{j+1}}) \iso x_{j+1} \downarrow\). More specifically we show that \(E_{j+1}\setminus E_j = \{e_{j+1}\}\) and that we can extend the bijections such that
	\begin{align}
		\label{eq:bij_te}
		\te{j+1} = \te{j}\cup\{(R_j\redl{i,\alpha} R_{j+1}) \mapsto e\}       \\
		\label{eq:bij_mt}
		\mt{j+1} = \mt{j}\cup\{e_{j+1} \mapsto (R_j\redl{i,\alpha} R_{j+1})\} \\
		\label{eq:bij_me}
		\me{j+1} = \me{j}\cup\{e_{j+1} \mapsto e\}
	\end{align}
	As \(\me{j+1}\) is a label and order preserving bijection on \(\xmax\), the maximal configuration in \(\encm{R_{j+1}}\), it follows that \(\me{j+1}(\xmax) = x_{j+1}\). As \(\encm{R_{j+1}} \iso \xmax\downarrow\) it follows \(\functoric(\encm{R_{j+1}}) \iso x_{j+1} \downarrow\).

	We now proceed by cases on the transition \(R_j\fwlts{i}{\alpha} R_{j+1}\).
	We distinguish two cases: \(\alpha\) is an unsynchronized input or output, and \(\alpha = \tau\).
	\begin{itemize}
		\item Let us suppose w.l.o.g. that \(\alpha = a\). Then we can rewrite the transition as follow:
		      \begin{align*}
			      R_j = (\namelist{b})(S\mid m\rhd a.P + Q) \redl{i,\alpha} R_{j+1} = (\namelist{b})(S\mid \mem{i, \alpha, Q}.m\rhd P)
		      \end{align*}

		      Let us also define the following projection on events: \(\pi_S (e) = e_s \) if there exists \(e_s\in \encm{S} \) such that \(e,e_s\) have the same identifiers and undefined otherwise. Similarly define \(\pi_m\) for the projections of events from \(\encm{R_i}\) to \(\encm{m}\).

		      Note that we can extend \(\te{j+1}\) as in~\autoref{eq:bij_te}.
		      We show that there exists an event in \(\encm{R_{j+1}}\) which corresponds to the transition \(R_j\redl{i,\alpha} R_{j+1}\).
		      Let us unfold the encoding of the two processes above:
		      \begin{align*}
			      \encm{R_j} = \encm{(\namelist{b})(S\mid m\rhd a.P + Q)} = \big( r\circ (\encm{S} \times \encm{m})\big)\restr{\bot} \restr{\namelist{b}} \\
			      \encm{R_{j+1}} =\encm{(\namelist{b})(S\mid \mem{i, \alpha, Q}.m\rhd a)} = \big( r\circ (\encm{S} \times \encm{\mem{i, \alpha, Q}.m}) \big)\restr{\bot} \restr{\namelist{b}}.
		      \end{align*}
		      Let us write \(\encm{m} = (E_m,C_m, \labl_m, \iden_m)\). From~\autoref{lem:prefix_enc} we have that there exists a single maximal configuration in \(C_m\), denoted with \(\xmax^m\). Using~\autoref{def:encode_m} we can unfold \(\encm{\mem{i, \alpha, Q}.m}\) and write
		      \begin{align*}
			      \encm{R_j}     & = \big(r\circ (\encm{S} \times (E_m,C_m, \labl_m,\iden_m))\big)\restr{\bot} \restr{\namelist{b}}             \\
			      \encm{R_{j+1}} & = \big( r\circ (\encm{S} \times                                                                              \\
			                     & \qquad\qquad(E_m\cup\{e_m\},C_m\cup (\xmax^m\cup\{e_m\}),                                                    \\
			                     & \qquad\qquad\labl_m \cup \{e_m \to \alpha\}, \iden_m + \{e_m\to i\})) \big)\restr{\bot} \restr{\namelist{b}}
		      \end{align*}
		      for some event \(e_m\notin E_m\).
		      From rules act.\ and par.\ of~\autoref{ltsrules}, \(i\) is not in the domain of \(\iden_m\).
		      Therefore all synchronisations in \(\encm{S} \times \encm{\mem{i, \alpha, Q}.m}\) of the form \((e_s, e_m)\), with \(e_s\in\encm{S}\), are relabeled \(\bot\) and removed by the first restriction. The event \((\star,e_m)\) is preserved by the first restriction. It is not removed by the second restriction as \(a\notin\namelist{b}\).
		      Remember that from~\autoref{lem:prefix_enc} we have that there exists a single maximal configuration in \(\encm{R_{j+1}}\) and from the definition of the parallel composition there is only one event in \(xmax\) with the first projection equal to \(e_m\), denoted \(e_{j+1}\): \(e_{j+1}\in\xmax\) with \(\pi_m(e_{j+1})=e_m\).
		      It follows that \(E_{j+1} = E_j \cup\{e_{j+1}\}\). We extend then the bijections as in Equations~\ref{eq:bij_mt} and \ref{eq:bij_me}.
		      Moreover, \(\labl(e_{j+1})=\labl_m(e_m)\), and from~\autoref{eq:op_corr} it follows that \(\me{j+1}\) is label preserving.

		      The last part is to show~\autoref{eq:causality_max}.
		      We only have to show that
		      \begin{align}
			      \label{eq:causality_rj}
			      e'\leq_{\xmax} e_{j+1} \iff\me{j+1}(e')\leq_{x_{j+1}}\me{j+1}(e_{j+1}),
		      \end{align}
		      as the rest follows by induction on \(\me{j}\) and from~\autoref{eq:bij_me}. To show \autoref{eq:causality_rj}, consider \(e'\leq_{\xmax} e_{j+1}\). From~\autoref{def:encode_m} there exists a memory event \(d = \mem{i',\alpha',P'}\) in \(\encm{R_{j+1}}\) such that \(\iden_{j}(e') = i'\). From~\autoref{lem:causality_mem} we have then that \(i'\leq_{R} i\) and using the definition of causality on transitions (\cite[Definition 1]{Danos2005}), \(\mt{j+1}(e') \leq_{\theta} \mt{j+1}(e_{j+1})\).

		      We conclude using~\autoref{lem:te} which shows that \(\te{j+1}(\mt{j+1}(e')) \leq_{x_{j+1}} \te{j+1}(\mt{j+1}(e_{j+1}))\).

		      Similarly, we reason for \(e'\) concurrent with \(e_{j+1}\). Lastly, note that there are no events in conflict with \(e_{j+1}\) (or with \(e\) in \(x_{j+1}\downarrow\)), as there is a single maximal configuration in both \(\encm{R_{j+1}}\) and in \(x_{j+1}\downarrow\).

		\item Suppose that \(\alpha = \tau\) and let us write the transitions as follows:
		      \begin{align*}
			      R_j & = (\namelist{b})(S\mid m_1\rhd a.P_1 + Q_1\mid m_2\rhd \out{a}.P_2 + Q_2)                                       \\
			          & \redl{i,\tau} R_{j+1} = (\namelist{b})(S\mid \mem{i, a, Q_1}.m_1\rhd P_1\mid \mem{i, \out{a}, Q_2}.m_2\rhd P_2)
		      \end{align*}
		      We are assuming here, for simplification, that both thread involved in the synchronisation are under the same set of restricted names. The more general case, does not change the reasoning here, just adds in technicality.

		      We show that the transition adds a single event \(e_{j+1}\) in \(\encm{R_{j+1}}\) and that the bijection \textsc{me} defined on \(\encm{R_j}\) extends to \(e_{j+1}\) such that it remains a label and order preserving bijection between the maximal configuration in \(\encm{R_{j+1}}\) and \(x_{j+1}\). The proof follows the reasoning above.
	\end{itemize}
	\qed
\end{proof}

\subsection{Proof of \autoref{th:hpb_rccs}}%
\label{subsec:apC}
Before proving the main theorem, let us make the following observation. Let \(f:(E_1,C_1,\labl_1)\to(E_2,C_2,\labl_2)\) be a morphism, which is defined as function on events: \(f=(f_E,f_C)\). As \(f_C\) is defined by \(f_E\), we can w.l.o.g. write \(f=f_E\).
Saying that \(f\) is an isomorphism is equivalent then to saying that \(f: E_1 \rightharpoonup E_2\) is a label- and order-preseving bijection.
Therefore the functions \(f\) of \autoref{def:hpb} and \autoref{def:whpb_rccs} are of \enquote{the same nature}.

\begin{proof}[\autoref{th:hpb_rccs}, page~\pageref{th:hpb_rccs}]
	Let us prove the \HHPB case, the other three cases being similar, and actually simpler.

	\begin{description}
		\item[\(\Rightarrow\)]
		      Let \(\rel_{\text{RCCS}}\) be a \HHPB between \(\emptyset\rhd P_1\) and \(\emptyset\rhd P_2\).
		      We show that the following relation
		      \[
			      {\rel} =
			      \begin{multlined}[t]
				      \{(x_1,x_2,f) \setst x_1\in\enc{P_1}, x_2\in\enc{P_2}, \exists R_1, R_2\text{ \st }\orig{R_1}=P_1,\\
				      \orig{R_2} = P_2, (R_1,R_2,f)\in \rel_{\text{RCCS}} \text{ and }\\
				      \enc{R_1}=(\enc{P_1},x_1), \enc{R_2}=(\enc{P_2},x_2)\}
			      \end{multlined}
		      \]

		      is a \HHPB between \(\enc{P_1}\) and \(\enc{P_2}\).

		      First note that \((\emptyset,\emptyset,\emptyset)\in\rel\): indeed \((\emptymem\rhd P_1,\emptymem\rhd P_2,\emptyset)\in \rel_{\text{RCCS}}\) and \(\enc{\emptymem\rhd P_i} = (\enc{P_i}, \emptyset)\), for \(i\in\{1,2\}\).

		      Let us suppose that \((x_1,x_2,f)\in\rel\) for \(\enc{R_i}=(\enc{P_i},x_i)\), for \(i\in\{1,2\}\) and \(f:x_1\to x_2\) an isomorphism.
		      Moreover, note that \(\functoric(\encm{R_i}) \iso x_i\downarrow\), from \autoref{lem:enc_mem_2}, and that \(\encm{R_i} \iso x_i\downarrow + \iden\), for some function \(\iden\), from \autoref{def:iconf}.

		      To show that \(\rel\) is a \HHPB we have to show that if \(x_1 \redl{e_1} y_1\) (or \(x_1\revredl{e_1} y_1\)) then there exists \(y_2\) such that \(x_2 \redl{e_1} y_2\) (or \(x_2\revredl{e_2} y_2\) respectively) and such that \((y_1, y_2, f')\in\rel\) for some \(f'\).

		      Let \(x_1\redl{e_1}y_1\), hence by definition, \(y_1 = x_1\cup\{e_1\}\). From the correspondence between RCCS and their encodings (from~\cite[Lemma 6]{Aubert2016jlamp}), it follows that \(R_1\fwlts{i}{\alpha} S_1\) such that \(\enc{S_1} = (\enc{P_1},y_1)\). We therefore deduce that \(\encm{S_1} \iso y_1\downarrow + (\iden\cup \{e_1\mapsto i\})\).

		      As \((R_1,R_2,f)\in \rel_{\text{RCCS}}\) and as \(R_1\fwlts{i}{\alpha} S_1\), it follows that there exists a transition \(R_2\fwlts{j}{\alpha} S_2\) with \(f = f'\restr{\functoric(\encm{R_1})}\) and \((S_1,S_2,f')\in\rel_{\text{RCCS}}\).

		      Again from the correspondence between \(R_2\) and \(\enc{R_2}\) we have that \(x_2\redl{e_2}y_2\) such that \(y_2=x_2\cup\{e_2\}\) and \(\enc{S_2} = (\enc{P_2}, y_2)\). Then we have that \((y_1,y_2,f')\in \rel\).

		      We treat similarly the cases where \(x_2\) does a transition, or when the transitions are backwards.

		\item[\(\Leftarrow\)]
		      Let \(\rel_{\text{CONF}}\) be a \HHPB between \(\enc{P_1}\) and \(\enc{P_2}\).
		      We show that the following relation
		      \begin{align*}
			      {\rel} =
			      \begin{multlined}[t]
				      \{(R_1,R_2,f) \setst \orig{R_1}=P_1, \orig{R_2} = P_2\text{ and }\enc{R_1}=(\enc{P_1},x_1),\\
				      \enc{R_2}=(\enc{P_2},x_2),\text{ with }(x_1,x_2,f)\in\rel_{\text{CONF}} \}
			      \end{multlined}
		      \end{align*}
		      is a \HHPB between \(\emptyset\rhd P_1\) and \(\emptyset\rhd P_2\).

		      We have that \((\emptymem\rhd P_1,\emptymem\rhd P_2,\emptyset)\in \rel\) as \((\emptyset,\emptyset,\emptyset)\in\rel_{\text{CONF}}\) and \(\enc{\emptymem\rhd P_i} = (\enc{P_i}, \emptyset)\), for \(i\in \{1,2\}\).

		      We suppose now that \((R_1,R_2,f)\in\rel\), with \(f:\functoric(\encm{R_1}\to\functoric(\encm{R_2}))\). It implies that \(\enc{R_i}=(\enc{P_i},x_i)\), \(i\in \{1,2\}\), we have that \((x_1,x_2,f)\in\rel_{\text{CONF}}\). As \(\functoric(\encm{R_i}) \iso x_i\downarrow\), from \autoref{lem:enc_mem_2},\(f\) is also defined from \(x_1\) to \(x_2\).

		      To show that \(\rel\) is a \HHPB we have to show that if \(R_1\fwlts{i}{\alpha} S_1\) (or \(R_1\bwlts{i}{\alpha} S_1\)) then there exists \(S_2\) such that \(R_2\fwlts{j}{\alpha} S_2\) (or \(R_2\bwlts{j}{\alpha} S_2\) respectively) and such that \((S_1,S_2,f')\in\rel\) for some \(f'\).

		      Let \(R_1\fwlts{i}{\alpha} S_1\). We use again the correspondence between RCCS and their encodings (from~\cite[Lemma 6]{Aubert2016jlamp}) from which we have that there exists \(e_1\) and \(y_1=x_1\cup\{e_1\}\) such that \(x_1\redl{e_1}y_1\) and \(\enc{S_1} = (\enc{P_1},y_1)\). As \((x_1,x_2,f)\in\rel_{\text{CONF}}\) it implies that there exists \(e_2\), \(y_2\) and \(f'\) such that \(x_2\redl{e_2}y_2\) and \((y_1,y_2,f')\in \rel_{\text{CONF}}\). Again, from the correspondence between RCCS and configuration structures we have that from \(x_2\redl{e_2}y_2\), there exists \(S_2\) and some \(j\) such that \(R_2\fwlts{j}{\alpha} S_2\) with \(\enc{S_2} = (\enc{P_2}, y_2)\). Hence \(\functoric(\encm{S_2}) = y_2\downarrow\). We conclude therefore that \((S_1,S_2,f')\in \rel\).

		      Similarly we show the cases where \(R_1\) does a backward transition, or if \(R_2\) does a forward or backward transition.

	\end{description}
	\qed
\end{proof}

\end{document}